\documentclass[runningheads,11points]{article}

\usepackage{amsfonts, amsmath, amssymb}
\usepackage[affil-sl]{authblk} 

\usepackage[ruled,vlined]{algorithm2e}

\usepackage{tikz,tkz-graph}

\usepackage{amssymb}
\usepackage{cite,scrtime}
\usepackage{graphics}
\usepackage{amsmath}
\usepackage{mathrsfs}
\usepackage{verbatim}
\usepackage{boxedminipage}
\usepackage{colortbl}
\usepackage[breaklinks, bookmarks=false]{hyperref}

\newcommand{\TRUE}[1]{\textsc{True}}
\newcommand{\FALSE}[1]{\textsc{False}}

\usepackage{graphicx}

\usepackage{vmargin}

\setmarginsrb{3.9cm}{1.2cm}{3.9cm}{3.cm}{1.2cm}{1.2cm}{0.8cm}{1.4cm}

\usepackage{booktabs} 
\usepackage{boxedminipage}

\newtheorem{theorem}{Theorem}
\newtheorem{corollary}{Corollary}
\newtheorem{definition}{Definition}
\newtheorem{lemma}{Lemma}

\newenvironment{hey}[1]{\noindent{\bf #1.}\addcontentsline{toc}{subsection}{#1}}{\vspace{2mm}}

\newcommand{\rtcp}[1]{\hfill\tcp{#1}}

\newcommand\divides{|}

\newcommand{\PDI}{\textsc{PDI}\xspace}
\newcommand{\PDIfull}{\textsc{Planar Diameter Improvement}\xspace}
\newcommand{\OPDIfull}{\textsc{Outerplanar Diameter Improvement}\xspace}
\newcommand{\OPDI}{\textsc{OPDI}\xspace}

\newcommand{\dist}[2][]{\ensuremath{\operatorname{dist}}\ifx\relax#1\relax\else\ensuremath{_{#1}}\fi\ensuremath{(#2)}}
\newcommand{\degree}[2][]{\ensuremath{\operatorname{deg}}\ifx\relax#1\relax\else\ensuremath{_{#1}}\fi\ensuremath{(#2)}}
\newcommand{\diam}[1]{\ensuremath{\operatorname{diam}(#1)}}

\newcommand{\pdist}[2][]{\ensuremath{\operatorname{pdist}}\ifx\relax#1\relax\else\ensuremath{_{#1}}\fi\ensuremath{(#2)}}

\newtheorem{obs}{Observation}
\newtheorem{claimN}{Claim}
\newtheorem{fact}{Fact}

\newenvironment{proof}[1][]{\par \noindent {\bf Proof:#1}\ }{\hfill$\Box$\\}

\title{A Polynomial-time Algorithm for\\Outerplanar Diameter Improvement\thanks{Research supported by the Languedoc-Roussillon Project ``Chercheur d'avenir'' KERNEL and the French project EGOS (ANR-12-JS02-002-01).  The third author was co-financed by the European Union (European Social Fund ESF) and Greek national funds through the Operational Program ``Education and Lifelong Learning'' of the National Strategic Reference Framework (NSRF) - Research Funding Program: ARISTEIA II. Emails of authors: 
\href{mailto:nathann.cohen@gmail.com}{\sf nathann.cohen@gmail.com},
\href{mailto:goncalves@lirmm.fr}{\sf goncalves@lirmm.fr},
\href{mailto:eunjungkim78@gmail.com}{\sf eunjungkim78@gmail.com},
\href{mailto:Christophe.Paul@lirmm.fr}{\sf Christophe.Paul@lirmm.fr},
\href{mailto:Ignasi.Sau@lirmm.fr}{\sf Ignasi.Sau@lirmm.fr},
\href{mailto:sedthilk@thilikos.info}{\sf sedthilk@thilikos.info},
\href{mailto:weller@lirmm.fr}{\sf weller@lirmm.fr}\ .}
}

\author[1]{Nathann Cohen}
\author[2]{Daniel Gon\c{c}alves}
\author[3]{Eun Jung  Kim}
\author[2]{Christophe Paul}
\author[2]{Ignasi Sau}
\author[2,4]{Dimitrios M. Thilikos}
\author[2]{Mathias Weller}

\affil[1]{\emph{\small CNRS, LRI, Orsay, France}}
\affil[2]{\emph{\small CNRS, LIRMM, Montpellier, France}}
\affil[3]{\emph{\small CNRS, LAMSADE, Paris, France}}
\affil[4]{\emph{\small Department of Mathematics, University of Athens, Athens, Greece}}

\begin{document}

\maketitle
\setcounter{footnote}{0}

\vspace{-.35cm}

\begin{abstract}
\noindent The \OPDIfull problem asks, given a graph $G$ and an integer $D$, whether it is possible to add edges to $G$ in a way that the resulting graph is outerplanar and has diameter at most $D$. We provide a dynamic programming algorithm that solves this problem in polynomial time. \OPDIfull demonstrates several structural analogues to the celebrated and challenging \PDIfull  problem, where the resulting graph should, instead, be planar. The complexity status of this latter problem is open.
 \vspace{0.25cm}


\end{abstract}

\section{Introduction}
\label{sec:intro}

\looseness=-1
In general, a {\em graph completion problem} asks whether
it is possible to add edges to some given input graph so that the resulting graph satisfies some target property.
There are two different ways of defining the optimization measure for such problems. The first, and most common,
is the number of edges to be added, while the second is
 the value of some graph invariant on the resulting graph.
Problems of the first type are {\sc Hamiltonian Completion}~\cite{GareyJ79comp},  {\sc Interval Graph Completion}~\cite{HeggernesPTAV07}, {\sc Proper Interval Graph Completion}~\cite{GolumbicKS94onth,KaplanST99tract}, {\sc  Chordal Graph Completion}~\cite{Yannakakis81comp,KaplanST99tract}, and {\sc  Strongly Chordal  Graph Completion}~\cite{KaplanST99tract}, where the property is being Hamiltonian, interval, proper interval, chordal, and strongly chordal, respectively.

We focus our attention to the second category of problems where, for some given parameterized
 graph property ${\cal P}_{k}$, the problem asks, given a graph $G$ and an integer $k$, whether it is
 possible to add  edges to $G$ such that the resulting graph belongs to ${\cal P}_{k}$.
Usually ${\cal P}_{k}$ is a parameterized graph class whose graphs are typically required (for every $k$)
to satisfy some sparsity condition. There are few problems of this type in the bibliography.
Such a completion  problem  is the {\sc Planar Disjoint Paths Completion} problem that asks,
given a plane graph and a collection of $k$ pairs of terminals, whether it is possible to add edges
such that the resulting graph remains  plane and contains $k$ vertex-disjoint paths between the pairs of terminals.
While this problem is trivially NP-complete,  it has been
studied from the point of view of parameterized complexity~\cite{AdlerKT11plan}.
In particular, when all edges should be added in the same face, it can be solved in $f(k)\cdot n^{2}$ steps~\cite{AdlerKT11plan}, i.e., it is fixed parameter tractable ({\sf FPT} in short).

Perhaps the most challenging problem of the second category is the \PDIfull problem (\PDI in short), which was first mentioned by Dejter and Fellows~\cite{DF93} (and made an explicit open problem in~\cite{DF99}). Here we are given
a planar (or plane) graph and we ask for the minimum integer $D$ such that it is possible to add edges so that
the resulting graph is a planar graph with diameter at most $D$ (according to the general formalism,
for each $D$, the parameterized property ${\cal P}_{D}$ contains all planar graphs with diameter at most $D$). The computational complexity of \PDIfull  is open, as it is not even known whether it is an {\sf NP}-complete problem, even in the case where the embedding is part of the input. Interestingly, \PDIfull is  known to be {\sf FPT}:
it is easy to verify that, for every $D$, its \textsc{Yes}-instances are closed under taking minors\footnote{To see this, if a graph $G$ can be completed into a planar graph $G'$ of diameter $D$, then $G'$ is also a valid completion of any subgraph $H\subseteq G$. Similarly, by merging two adjacent vertices $uv$ in both $G$ and $G'$, one sees that the diameter is also closed under edge contraction.}
which, according to the meta-algorithmic consequence of the Graph Minors series of Robertson and Seymour~\cite{RobertsonS95-XIII,RobertsonS04-XX}, implies that \PDIfull is {\sf FPT}. Unfortunately, this implication only
proves the {\sl existence} of such an algorithm for each~$D$,
while it does not give any way to construct it. A uniform FPT
algorithm for this problem  remains as one of the most intriguing open problems in parameterized algorithm design.
To our knowledge, when it comes to
explicit algorithms, it is not even clear how to get an $O(n^{f(D)})$-algorithm for this problem
(in parameterized complexity terminology, such an algorithm is called an {\sf XP}-algorithm).

\looseness=-1
Notice that, in both aforementioned problems of the second type, the planarity of the graphs in~${\cal P}_{D}$
is  an important restriction, as it is essential for generating a non-trivial problem; otherwise, one could immediately
turn a graph into a clique that trivially belongs to~${\cal P}_{1}$. For practical purposes, such problems
are relevant where instead of generating few additional links, we mostly care about maintaining the network topology.
The algorithmic and graph-theoretic study  on diameter improvement problems has focused both on
the case of minimizing the number (or weight) of added edges~\cite{AlonGR00,BiloGP12,ChepoiENV06,DodisK99,ErdosGR98,Ishii13},
as well as on the case of minimizing  the diameter~\cite{BiloGP12,FratiGGM13}.
 In contrast, the network topology, such as acyclicity or planarity, as a constraint to be preserved has received little attention in the context of complementing a graph; see for example~\cite{ErdosGR98}. See also~\cite{Kant96,JBCR14} for other augmentation problems in outerplanar graphs, where the objective is to add edges in order to achieve a prescribed connectivity.

 \smallskip

In this paper we study the  \OPDIfull problem, or \OPDI in short. An instance of \OPDI consists of an outerplanar graph $G=(V,E)$ and a positive integer $D$, and we are asked to add a set $F$ of missing edges to $G$ so that the resulting graph $G'=(V,E\cup F)$ has diameter at most $D$, while $G'$ remains  outerplanar. Note that we are allowed to add arbitrarily many edges as long as the new graph is outerplanar. Given a graph $G=(V,E)$, we call $G'=(V,E\cup F)$ a {\em completion} of $G$.

\looseness=-1
It appears that the combinatorics of \OPDI demonstrate some interesting parallelisms with the notorious \PDI problem.
We denote by ${\bf opdi}(G)$ (resp. ${\bf pdi}(G)$) the minimum diameter of an outerplanar (resp. planar) completion of $G$.
It can be easily seen that the treewidth of a graph with bounded  ${\bf pdi}(G)$ is bounded, while
the pathwidth of a graph with bounded  ${\bf opdi}(G)$ is also bounded.
%
In that sense, the \OPDI can be seen as the ``linear counterpart''  of \PDI. We stress that the same ``small pathwidth'' behavior of \OPDI holds even  if, instead of outerplanar graphs, we consider any class of graphs with bounded outerplanarity.
Note also  that both ${\bf pdi}(G)$ and ${\bf opdi}(G)$ are trivially $2$-approximable in the particular case where the embedding is given.
To see this, let $G'$ be a triangulation of  a plane (resp. outerplane) embedding of $G$ where, in every face of~$G$, all edges added to it have a common endpoint. It easily follows  that, for both graph invariants, the diameter of $G'$ does not exceed twice the optimal value.


%
%
%
%

Another closely related notion is that of {\em $t$-spanners}. Given a graph $G=(V,E)$, a {\em $t$-spanner} of $G$ is a spanning subgraph $G'=(V,E')$ such that for every pair of vertices, the distance in $G'$ is within a factor $t$ to their distance in $G$. The measure $t$ for the quality of a spanner is called the {\em stretch factor}. The problem of identifying a $t$-spanner, or the {\sc $t$-Spanner} problem, has been extensively studied in geometry as well as in communication network design  \cite{WangL06,CaiC95,ChandraDNS95,Cohen98,DraganFG08}. A sparse or {\sl planar} $t$-spanner is of particular importance: in wireless ad-hoc networks, certain routing protocols
require the network topology to be planar \cite{WangL06}.
The notion of $t$-spanner is, in a sense, dual to
  that of \PDI or \OPDI. The former allows losing the established links and the latter allows having additional links while, in both cases, our interest lies in preserving the constraint on the network topology. Note that in both cases, we want to improve (or to not deteriorate a lot) the performance of the network, which is measured by the lags in information transfer, such as diameter or stretch factor.

\vspace{.15cm}
\noindent \textbf{Our results}. In this work, we show that  \OPDIfull is polynomial-time solvable. Our algorithm, described in Section~\ref{sec:DPalgo}, is based on dynamic programming and works in full generality, even when the input graph may be disconnected. Also, our algorithm does {\sl not} assume
that the input comes with some specific embedding (in the case of an embedded input, the problem becomes considerably easier to solve).


\vspace{.15cm}
\noindent \textbf{Further research}. Our algorithm for  \OPDI runs in time $O(n^3)$ for connected input graphs, and in time $O(n^{7})$ or $O(n^{9})$ for disconnected input graphs, depending on whether $D$ is odd or even, respectively, which can probably be improved.

We believe that our approach might be interesting for generalizations or variations of the \OPDI  problem, such as the one where we demand that the completed graph has fixed outerplanarity or is series-parallel.

By the Graph Minors series of Robertson and Seymour~\cite{RobertsonS95-XIII,RobertsonS04-XX}, we know that for each fixed integer $D$, the set of \emph{minor obstructions} of \OPDI is \emph{finite}. We have some preliminary results in this direction, but we managed to obtain a complete list only for small values of $D$. Namely, we obtained a partial list of forbidden substructures (not necessarily minimal), by using the notion of \emph{parallel matching}. These partial results can be found in Appendix~\ref{sec:maxParallelMatching}.

Settling the computational complexity of
\PDI  remains the main open problem in this area. An explicit {\sf FPT}-algorithm, or even an {\sf XP}-algorithm, would also be significant. Again, we have some partial {\sf NP}-completeness result in this direction, but for a problem which is slightly more complicated than \PDI, in particular involving edge weights. This reduction can be found in Appendix~\ref{sec:NPh}.




\section{Description of the algorithm}
\label{sec:DPalgo}

The aim of this section is to describe a polynomial-time dynamic program that, given an outerplanar graph $G$ and an integer $D$, decides whether $G$ admits an outerplanar completion with diameter at most $D$, denoted \emph{diameter-$D$ outerplanar completion} for simplicity. Note that such an algorithm easily yields a polynomial-time algorithm to compute, given an outerplanar graph $G$, the smallest integer $D$ such that $G$ admits a  diameter-$D$ outerplanar completion.

Before describing the algorithm, we show some properties of outerplanar completions. In particular, Subsection~\ref{subsec:cutvertices} handles the case where the input outerplanar graph has cut vertices. Its objective is to prove that we can apply a {\sl reduction rule} to such a graph which is safe for the \OPDI problem. In Subsection~\ref{subsec:2vertexSep} we deal with 2-vertex separators, and in Subsection~\ref{subsec:dp-prog} we present a polynomial-time algorithm for {\sl connected} input graphs. Finally, we present the algorithm for disconnected input graphs in Subsection~\ref{subsec:disconnected}.


\paragraph{\textbf{Some notation}.} We use standard graph-theoretic notation, see for instance~\cite{Diestel05}. It is well known that a graph is outerplanar if and only if it excludes $K_4$ and $K_{2,3}$ as a minor. An outerplanar graph is \emph{triangulated} if all its inner faces (in an outerplanar embedding) are triangles.  An outerplanar graph is
\emph{maximal} if it is 2-connected and triangulated.  Note that, when solving the \OPDI problem, we may always assume that the completed graph $G'$ is maximal.



\subsection{Reducing the input graph when there are cut vertices}
\label{subsec:cutvertices}

\newcommand{\ecc}{\ensuremath{\operatorname{ecc}}}
\newcommand{\bestECC}{\ensuremath{\ecc^*}}
\newcommand{\bestECCK}{\ensuremath{\ecc_D^*}}

Given a graph $G$, let the \emph{eccentricity} of a vertex $u$ be
$\ecc(u,G) = \max_{v\in V(G)}\dist[G]{u,v}$.  Given an outerplanar
graph $G$, a vertex $u\in V(G)$, and an integer $D$, let us define $\bestECCK(u,G)$ as
$\min_{H} \ecc(u,H)$, where the minimum is taken over all the
diameter-$D$ outerplanar completions $H$ of $G$. If all the
outerplanar completions have diameter more than $D$, we set this
value to $\infty$. Unless said otherwise, we assume henceforth that $D$ is a fixed given integer, so we may just write $\bestECC(u,G)$ instead of $\bestECCK(u,G)$. (The value of $D$ will change only in the description of the algorithm at the end of Subsection~\ref{subsec:dp-prog}, and in that case we will make the notation explicit).


As admitting an outerplanar completion with bounded eccentricity is a minor-closed property, let us observe the following:
\begin{lemma}\label{lem-ecc}
For any connected outerplanar graph $G$, any vertex $v\in V(G)$, and
any connected subgraph $H$ of $G$ with $v\in V(H)$, we have that
$\bestECC(v,H) \le \bestECC(v,G)$.
\end{lemma}
\begin{proof}
Let $G'$ be an outerplanar completion of $G$ achieving $\bestECC(v,G)$, that is
such that $\ecc(v,G') = \bestECC(v,G)$. Contracting  the edges of $G'$ that have
at least one endpoint out of $V(H)$ one obtains an outerplanar completion $H'$ of
$H$ (as outerplanar graphs are minor-closed). As contracting an edge does not
elongate any shortest path, we have that $\dist[H']{v,u} \le \dist[G']{v,u}$ for any vertex $u\in V(H)$, and in particular the diameter of $H'$ is at most the diameter of
$G'$, so $\bestECC(v,H) < \infty$. We thus have that $\bestECC(v,H) \le \ecc(v,H') \le \ecc(v,G') = \bestECC(v,G)$.
\end{proof}

\newcommand{\branches}[2][1]{\ensuremath{B_{#1\ldots #2}}}

Consider a connected graph $G$ with a cut vertex $v$, and let $C_1,\ldots ,C_t$ be the vertex sets of the connected components of $G\setminus \{v\}$. For $1\le i\le t$, we call the vertex set $B_i = C_i\cup \{v\}$ a \emph{branch} of $G$ at $v$.
To shorten notations, we abbreviate $B_i\cup\ldots\cup B_j=:\branches[i]{j}$, for any $1\le i < j \le t$.
Also, when referring to the eccentricity, we simply denote $G[B_i]$ by $B_i$. Thus, for example, when considering the value $\bestECC(v,\branches{i})$, it will refer to the minimum eccentricity with respect to $v$ that a diameter-$D$ outerplanar completion of the graph $G[\branches{i}]$ can have.

The following lemma, which is crucial in order to obtain a polynomial-time algorithm, implies that if $G$ has a cut vertex $v$
with many branches, it is safe to remove most of them.

\begin{lemma}\label{lem-7-blocks}
  Consider an outerplanar graph $G$ with a cut vertex $v$ that
  belongs to at least 7~branches. Denote these branches
  $B_1,\ldots,B_t$, with $t\ge 7$, in such a way that $\bestECC(v,B_1) \ge
  \bestECC(v,B_2) \ge \ldots \ge \bestECC(v,B_t)$.  The graph $G$ has an outerplanar completion with diameter at
  most $D$ if and only if $\bestECC(v,\branches{6}) + \bestECC(v,B_7) \le D$.
\end{lemma}
\begin{proof}
  ``$\Leftarrow$'':
  If $\bestECC(v,\branches{6}) + \bestECC(v,B_7) \le D$, gluing on $v$
  the outerplanar completions of $G[\branches{6}], G[B_7],\ldots ,G[B_t]$, respectively achieving $\bestECC(v,\branches{6}), \bestECC(v,B_7), \ldots , \bestECC(v,B_t)$, one obtains a diameter-$D$ outerplanar completion $G'$ of $G$. Indeed,
  \begin{itemize}
    \item[$\bullet$] The graph obtained is outerplanar and contains $G$.
    \item[$\bullet$] Two vertices $x,y$ of $G[\branches{6}]$ (resp. of $G[B_i]$ for $7\le i\le t$) are at distance at most $D$ from each other, as $\bestECC(v,\branches{6}) < \infty$ (resp. as $\bestECC(v,B_i) < \infty$).
    \item[$\bullet$] Any vertex $x$ of $G[\branches{6}]$ and $y$ of $G[B_i]$, with $7\le i\le t$, are respectively at distance at most $\bestECC(v,\branches{6})$ and $\bestECC(v,B_i)\le \bestECC(v,B_7)$ from $v$. They are thus at distance at most $\bestECC(v,\branches{6}) + \bestECC(v,B_7) \le D$ from each other.
    \item[$\bullet$] Any vertex $x$ of $G[B_i]$ and $y$ of $G[B_j]$, with $7\le i<j\le t$, are respectively at distance at most $\bestECC(v,B_i)\le \bestECC(v,B_1) \le \bestECC(v,\branches{6})$ (By Lemma~\ref{lem-ecc}) and $\bestECC(v,B_j)\le \bestECC(v,B_7)$ from $v$. They are thus at distance at most $D$ from each other.
  \end{itemize}


  ``$\Rightarrow$'':
  In the following, we consider towards a contradiction an outerplanar graph $G$
  admitting a diameter-$D$ outerplanar completion,
  but such that
  \begin{equation}\label{eq:B1-6 + B7}
    \bestECC(v,\branches{6}) + \bestECC(v,B_7) > D.
  \end{equation}
  Among the triangulated diameter-$D$ outerplanar completions of $G$,
  let $G'$ be one that maximizes the number of branches at $v$.  Let
  $t'$ be the number of branches at $v$ in $G'$, and denote these
  branches $B'_1,\ldots,B'_{t'}$, in such a way that $\bestECC(v,B'_1)
  \ge \bestECC(v,B'_2) \ge \ldots \ge \bestECC(v,B'_{t'})$.
  Let~$S_{i'}:=\{i\mid B_i\subseteq B'_{i'}\}$ for all~$1\leq i'\leq
  t'$ (note that $\{S_1,\ldots,S_{t'}\}$ is a partition
  of~$\{1,\ldots,t\}$). The numbering of the $B'_{i'}$'s is such that if
  $\bestECC(v,B'_1) = \bestECC(v,B'_{i'})$ for some $i'>1$, then $\min
  S_1 < \min S_{i'}$.  Then, since $G'$ has diameter at most $D$ and
  shortest paths among distinct branches of $G'$ contain $v$, it is clear that
  \begin{equation}\label{eq:ecc among B'}
    \underset{1\leq i'<j'\leq t'}\forall\bestECC(v,B'_{i'}) + \bestECC(v,B'_{j'}) \leq D.
  \end{equation}
  The branches $B'_{i'}$ with $|S_{i'}|=1$ are called \emph{atomic}.

  \begin{claimN}\label{claim:split nonatomic}
    Let $B'_{i'}$ be a non-atomic branch and let $S'\subsetneq S_{i'}$. Then,
    $\bestECC(v,\bigcup_{i\in S'}B_i) + \bestECC(v, \bigcup_{i\in S_{i'} \setminus S'}B_i ) > D$.
  \end{claimN}
  \begin{proof}
    Let $\mathcal{B}:=\bigcup_{i\in S'} B_i$ and $\bar{\mathcal{B}}:=B'_{i'}\setminus\mathcal{B}$.
    If the claim is false, then $\bestECC(v,\mathcal{B})+\bestECC(v,\bar{\mathcal{B}})\le D$.
    Furthermore, for all~$j'\ne i'$,
    \[
      \bestECC(v,\mathcal{B}) + \bestECC(v,B'_{j'})
    \stackrel{\text{Lemma~\ref{lem-ecc}}}{\le}
      \bestECC(v,B'_{i'}) + \bestECC(v,B'_{j'})
    \stackrel{\eqref{eq:ecc among B'}}{\le} D\\
    \]
    and, likewise, $\bestECC(v,\bar{\mathcal{B}})+\bestECC(v,B'_{j'})\leq D$.
    Thus, the result of replacing $G'[B'_{i'}]$ with the disjoint union of an outerplanar completion achieving $\bestECC(v,\mathcal{B})$ and an outerplanar completion achieving $\bestECC(v,\bar{\mathcal{B}})$ yields a diameter-$D$ outerplanar completion containing more branches than $G'$, contradicting our choice of~$G'$.
  \end{proof}


  \begin{claimN}\label{claim:S consequtive}
    For some value $s\ge 1$, $S_1=\{1,\ldots,s\}$.
  \end{claimN}
  \begin{proof}
%
Towards a contradiction, assume that there is some $i\notin S_1$ with
$i+1\in S_1$. Let $i'>1$ be such that~$B_i\subseteq B'_{i'}$. Note
that $B'_1$ is not atomic, as otherwise $\bestECC(v,B'_1)
=\bestECC(v,B_{i+1}) \le \bestECC(v,B_{i}) \le \bestECC(v,B'_{i'})$,
contradicting the numbering of the $B'_j$'s.  Then,
    \begin{align*}
        \bestECC(v,B'_1\setminus(B_{i+1}\setminus v))) + \bestECC(v, B_{i+1})
      & \stackrel{\text{Lemma~\ref{lem-ecc}}}{\le}
        \bestECC(v,B'_1) + \bestECC(v, B_{i+1})\\
      & \makebox[11mm][c]{$\le$}
        \bestECC(v,B'_1) + \bestECC(v, B_i)\\
      & \makebox[11mm][c]{$\le$}
        \bestECC(v,B'_1) + \bestECC(v, B'_{i'})
      \stackrel{\eqref{eq:ecc among B'}}{\le}
        D,
    \end{align*}
    contradicting Claim~\ref{claim:split nonatomic}. Note that, as $S_1 \neq \emptyset$,  the claim follows.
  \end{proof}

  \begin{claimN}\label{claim:ecc-i-i+1}
    The integer $s$ is the least value such that
    $\bestECC(v,\branches{s}) + \bestECC(v,B_{s+1}) \le D$, and the
    largest value such that $\bestECC(v,\branches{s-1}) +
    \bestECC(v,B_s) > D$.
  \end{claimN}
  \begin{proof}
    Towards a contradiction, assume there is some $s'<s$ such that
    $\bestECC(v,\branches{s'}) + \bestECC(v,B_{s'+1}) \le D$. Then the
    graph obtained from the diameter-$D$ outerplanar completions of
    $\branches{s'}$ and $B_{i}$ for all $i>s'$, respectively achieving
    $\bestECC(v,\branches{s'})$ and $\bestECC(v,B_{i})$, would be a
    diameter-$D$ outerplanar completion of $G$ with more branches than
    $G'$, a contradiction.

    Again, assume towards a contradiction that there is some $s'>s$ such that
    $\bestECC(v,\branches{s'-1}) + \bestECC(v,B_{s'}) > D$.  By
    \eqref{eq:ecc among B'} and Lemma~\ref{lem-ecc}, $D\ge
    \bestECC(v,\branches{s}) + \bestECC(v,B_{s'})$, hence
    $\bestECC(v,\branches{s'-1}) > \bestECC(v,\branches{s})$. But this
    contradicts Lemma~\ref{lem-ecc}, as $\bestECC(v,\branches{s})=
    \ecc(v,G') \ge \bestECC(v,G)$.
  \end{proof}

By \eqref{eq:B1-6 + B7}, Claim~\ref{claim:ecc-i-i+1} implies that $s\ge 7$.

  \begin{claimN}\label{claim:vertex-ecc}
    Let~$S'\subseteq\{1,\ldots,t\}$ and let $\mathcal{B}:=\bigcup_{i\in S'}B_i$. Then, there is a vertex in~$\mathcal{B}$ that is, in $G'$, at distance at least $\bestECC(v,\mathcal{B})$ to any vertex of $V(G)\setminus(\mathcal{B}\setminus v)$.
  \end{claimN}
  \begin{proof}
Towards a contradiction, assume that for any vertex $u \in \mathcal{B}$ there exists a vertex $w \in V(G)\setminus(\mathcal{B}\setminus v))$ such that $\dist[G']{u,w} < \bestECC(v,\mathcal{B})$. From $G'$, contracting all vertices of $V(G)\setminus\mathcal{B}$ onto $v$ yields a graph $H$ with a path between $u$ and $v$ of length strictly smaller than $\bestECC(v,\mathcal{B})$. As this argument holds for any vertex $u \in \mathcal{B}$, it implies that $\ecc(v,H)<\bestECC(v,\mathcal{B})$. Since $H$ is an outerplanar completion of $G[\mathcal{B}]$, this contradicts the definition of~$\bestECC$.
  \end{proof}

  Two sub-branches $B_i$ and $B_j$ of $B'_1$ are \emph{linked} if $G'$
  has an edge intersecting both $B_i\setminus \{v\}$ and $B_j\setminus \{v\}$.

  \begin{claimN}\label{claim:>K-link}
    Let $1\le i< j \le s$ and let $\bestECC(v,\branches{i}) +
    \bestECC(v,B_j) > D$.  Then $\bestECC(v,\branches{i}) +
    \bestECC(v,B_j) = D+1$, and $B_{j}$ is linked to one of
    $B_1,\ldots,B_i$.
  \end{claimN}
  \begin{proof}
    By Claim~\ref{claim:vertex-ecc}, there is a vertex~$x\in B_j$ that is, in $G'$,
    at distance at least $\bestECC(v,B_j)$ to any vertex in $\branches{i}$. Likewise, there
    is a vertex $y\in\branches{i}$ that is, in $G'$, at distance at least $\bestECC(v,\branches{i})$
    to any vertex in $B_j$.
    Let $P$ be any shortest
    path of $G'$ between $x$ and $y$ (hence $P$ has length at most $D$). By construction, the maximal subpath of $P$ in
    $B_{j}\setminus v$ containing $x$ has length at least $\bestECC(v,B_j)-1$ and the
    maximal subpath of $P$ in $\branches{i}\setminus v$ containing $y$ has length at
    least $\bestECC(v,\branches{i})-1$.
    Since these subpaths are vertex disjoint the remaining part of $P$ has length $d_P\ge 1$.
    Hence $D\ge \bestECC(v,B_j) + \bestECC(v,\branches{i}) +d_P-2$.
    As $\bestECC(v,\branches{i}) + \bestECC(v,B_{j}) > D$, we have that $d_P=1$, and thus there is a single edge in $P$ linking $B_{j}$ and $\branches{i}$.
    This also yields to $\bestECC(v,B_j) + \bestECC(v,\branches{i}) = D+1$.
  \end{proof}

  Claim~\ref{claim:ecc-i-i+1} and Claim~\ref{claim:>K-link} clearly imply that for any $1\le i< s$, $B_{i+1}$ is linked to one of $B_1,\ldots,B_i$. A consequence of the next claim will be that $B_{i+1}$ is linked to exactly one of these branches.

  \begin{claimN}\label{claim:rel-path}
    The graph (with vertex set ${1,\ldots,s}$)
    induced by the relation \emph{linked} ($ij$ is an edge iff $B_i$ is linked to $B_j$) is a path.
    Furthermore, the subgraph induced by $\{1,\ldots,i\}$, for any $1\le i\le s$, is connected.
  \end{claimN}
  \begin{proof}
    The first statement is a consequence of the following three facts.
    \begin{enumerate}
      \item This graph is connected. Otherwise, $v$ would be a cut vertex in $G'[B'_1]$, contradicting the definition of branch.
      \item This graph has maximum degree 2. Consider for contradiction that some branch $B_i$ is linked to three branches $B_{j_1}$, $B_{j_2}$, and $B_{j_3}$. As each of $B_i\setminus v$, $B_{j_1}\setminus v$, $B_{j_2}\setminus v$, and $B_{j_3}\setminus v$ induces a connected graph in $G'$, these four sets together with $v$ induce a $K_{2,3}$-minor in $G'$, contradicting its outerplanarity.
      \item This graph is not a cycle. Assume otherwise, so in particular it implies that $s \ge 3$. As each $B_i\setminus v$ induces a connected graph in $G'$, these sets together with $v$ would induce a $K_4$-minor in $G'$, contradicting its outerplanarity.
    \end{enumerate}
    The second statement is a direct consequence of the fact that for any $1\le i< s$, $B_{i+1}$ is linked to one of $B_1,\ldots,B_i$.
  \end{proof}

  Hence, for any $1\le i\le s$, the graph $G'[\branches{i}\setminus v]$ is connected.

  \begin{claimN}\label{claim:ecc-increase}
    For any $3\le i < s$, $\bestECC(v,\branches{i}) > \bestECC(v,\branches{i-2})$.
  \end{claimN}
  \begin{proof}
    The monotonicity property given by Lemma~\ref{lem-ecc}  implies that $\bestECC(v,\branches{i})\geq  \bestECC(v,\branches{i-1}) \geq  \bestECC(v,\branches{i-2})$.
    Towards a contradiction suppose that $\bestECC(v,\branches{i}) = \bestECC(v,\branches{i-1}) = \bestECC(v,\branches{i-2}) =: c$. Then, Claim~\ref{claim:ecc-i-i+1} implies that $c + \bestECC(v,B_{j}) > D$ for all $j\in\{i-1,i,i+1\}$. Thus, by Claim~\ref{claim:>K-link}, each of $B_{i-1}$, $B_{i}$, and $B_{i+1}$
    is linked to one of $B_1,\ldots,B_{i-2}$. As each of $\branches{i-2}\setminus v$, $B_{i-1}\setminus v$, $B_{i}\setminus v$, and $B_{i+1}\setminus v$ induces a connected graph in $G'$, these sets together with vertex $v$ induce a $K_{2,3}$-minor, contradicting the outerplanarity of $G'$.
  \end{proof}

In the following let $q$ be any integer such that $3\le q\le s$ and $B_q$ is not linked to $B_1$.
Let $p<q$ be such that $B_p$ and $B_q$ are linked. By Claim~\ref{claim:rel-path}, $p$ is unique.


  As there is an edge between $B_p$ and $B_q$ and $G'$ is
  triangulated, the vertex $v$ has a neighbor $u \in B_p$ (in $G'$) that
  is adjacent to some vertex in $B_q\setminus v$; see Figure~\ref{fig:branches} for
  an illustration. 
  


  \begin{figure}[h!]
\vspace{-.65cm}
    \center\includegraphics[width=0.5\textwidth]{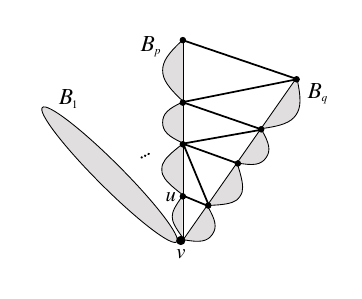}
    \vspace{-.65cm}
    \caption{\label{fig:branches}Structure of $G'[\branches{q}]$.}
\end{figure}

In an outerplanar graph any inner edge forms a $2$-separator. In $G'$
it is clear that $\{v,u\}$ separates $B_1\setminus v$ and
$B_q\setminus v$.  Let $(X,Y)$ be a separation  of $G'$ (that is, two sets $X,Y \subseteq V(G')$ such that $X \cup Y = V(G')$ and such that there are no edges between $X \setminus Y$ and $Y \setminus X$) such that
$X\cap Y =\{v,u\}$, $\branches{q-1}\setminus B_p\subsetneq X$  and
$B_q\subseteq Y$ (such a separation exists by Claim~\ref{claim:rel-path}).

  \begin{claimN}\label{claim:sameECC}
    $\bestECC(v,\branches{q}) = \bestECC(v,\branches{q-1})$.
  \end{claimN}
  \begin{proof}
    By Lemma~\ref{lem-ecc}, it suffices to show
    $\bestECC(v,\branches{q}) \le \bestECC(v,\branches{q-1})$.  To
    this end, let $H$ be the outerplanar completion of $\branches{q}$
    obtained from $G'$ by contracting every branch $B_i$, with $i>q$,
    onto $v$.  Since $H$ is a minor of $G'$, $H$ is a diameter-$D$
    outerplanar completion of $\branches{q}$. We show $\ecc(v,H)\le
    \bestECC(v,\branches{q-1})$.

    Consider any vertex $x\in X$, and let $y\in B_q \subsetneq Y$ be a
    vertex that is at distance at least $\bestECC(v,B_q)$ to both $v$
    and $u$ (such a vertex $y$ exists by Claim~\ref{claim:vertex-ecc}). As
    a shortest path between $x$ and $y$ (of length at most $D$) goes
    through $v$ or $u$, the vertex $x$ is thus at distance at most
    $D -\bestECC(v,B_q)$ to $v$ or $u$.  As $v$ and $u$ are
    adjacent, the vertex $x$ is at distance at most $D +1
    -\bestECC(v,B_q)$ ($= \bestECC(v,\branches{q-1})$ by
    Claim~\ref{claim:>K-link}) to $v$.  Hence every vertex $x\in X$ is
    at distance at most $\bestECC(v,\branches{q-1})$ to $v$ in $H$.

    Consider now any vertex $y\in Y\cap V(H)$, and let $x\in B_1
    \subsetneq X$ be a vertex that is at distance at least
    $\bestECC(v,B_1)$ to both $v$ and $u$ (such a vertex $x$ exists by
    Claim~\ref{claim:vertex-ecc}).  As a shortest path between $x$ and
    $y$ (of length at most $D$) goes through $v$ or $u$, the vertex
    $y$ is thus at distance at most $D -\bestECC(v,B_1)$ to $v$ or
    $u$. As $v$ and $u$ are adjacent, the vertex $y$ is at
    distance at most $D +1 -\bestECC(v,B_1)$ ($= \bestECC(v,B_2)$ by
    Claim~\ref{claim:>K-link}) to $v$. As $\bestECC(v,B_2) \le
    \bestECC(v,\branches{q-1})$ by Lemma~\ref{lem-ecc}, every vertex
    $y\in Y\cap V(H)$ is at distance at most
    $\bestECC(v,\branches{q-1})$ to $v$ in $H$.
  \end{proof}

We now claim that there exist two consecutive such values $q$ between $3$ and $6$. Indeed, note first that $B_1$ is linked to at most two other branches, as otherwise these branches together with $v$ and $B_1$ would induce a $K_{2,3}$-minor. Note also that by Claim~\ref{claim:ecc-i-i+1} and Claim~\ref{claim:>K-link}, $B_2$ is linked to $B_1$, so it follows that $B_1$ is linked to at most one branch $B_j$ with $j \geq 3$. Therefore, for $3 \leq q \leq 6$, there are at least two consecutive values of $q$ such that $B_q$ is not linked to $B_1$. Once we have these two consecutive values, say $i-1$ and $i$, we have by Claim~\ref{claim:sameECC} that $\bestECC(v,\branches{i-2}) = \bestECC(v,\branches{i})$, for some $i\le 6$, contradicting Claim~\ref{claim:ecc-increase}. This concludes the proof of the lemma.
\end{proof}

Our algorithm will compute the eccentricity of a given ``root'' vertex $r$ in a diameter-$D$ outerplanar completion~$G'$ of~$G$ in which this eccentricity is minimal, that is, $\bestECC(r,G)$. Then, however, the branch containing the root ($B_0$ in Algorithm~\ref{alg:connected}, Subsection~\ref{subsec:dp-prog}) should not be removed.
Therefore, although Lemma~\ref{lem-7-blocks} already implies that $G$ has a diameter-$D$ outerplanar completion if and only if $G[\branches{7}]$ does, we instead use the following corollary to identify removable branches.

\begin{corollary}\label{cor-7-blocks}
  Let $G$ be an outerplanar graph with a cut vertex $v$ that belongs
  to at least 8 branches. Denote these branches $B_1,\ldots,B_t$, with $t\ge
  8$, in such a way that $\bestECC(v,B_1) \ge \bestECC(v,B_2) \ge \ldots \ge
  \bestECC(v,B_t)$.  For each $8\le i\le t$, the graph $G_i = \bigcup_{j\in
  \{1,\ldots ,7,i\}} B_j$ has a diameter-$D$ outerplanar completion if and only if $G$ does.
\end{corollary}

\begin{proof}
  Recall that the property of having an outerplanar completion with bounded diameter is minor closed. Thus $G_i$ being a minor of $G$, we have that if $G$ admits a diameter-$D$ outerplanar completion, then so does $G_i$.

  On the other hand, if $G_i$ admits a diameter-$D$ outerplanar completion, by Lemma~\ref{lem-7-blocks} applied to $G_i$ we have that $\bestECC(v,\branches{6}) + \bestECC(v,B_7) \le D$. Thus gluing on $v$
  the outerplanar completions of $G[\branches{6}], G[B_7],\ldots ,G[B_t]$, respectively achieving $\bestECC(v,\branches{6}), \bestECC(v,B_7), \ldots , \bestECC(v,B_t)$, one obtains a diameter-$D$ outerplanar completion of $G$.
\end{proof}

\subsection{Dealing with 2-vertex separators}
\label{subsec:2vertexSep}

In this subsection, we extend the definition of eccentricity to the pairs $(u,v)$ such that $uv\in E(G)$. Namely, $\ecc(u,v,G)$ is defined as the set of pairs obtained by taking the maximal elements of the set $\{(\dist[G]{u,w},\dist[G]{v,w})\ |\ w\in V(G)\}$. The pairs are ordered such that $(d_1,d_2) \le (d'_1,d'_2)$ if and only if $d_1\le d'_1$ and $d_2\le d'_2$.  As $u$ and $v$ are adjacent, note that $\dist[G]{u,w}$ and $\dist[G]{v,w}$ differ by at most one. Hence, $\ecc(u,v,G)$ equal to one of $\{(d,d)\}$, $\{(d,d+1)\}$, $\{(d+1,d)\}$, and $\{(d,d+1),(d+1,d)\}$, for some positive integer $d$. Given a graph $G$ and a subset $S \subseteq V(G)$, we denote by $\partial(S)$ the set of vertices in $S$ that have at least one neighbor in $V(G) \setminus S$.

\begin{lemma}\label{claim:ecc-triangle}
Consider a connected graph $G$ with $V(G)=: X $ and a triangle $uvw$ and two sets $X_u, X_v \subseteq X$ such that $X_u \cup X_v = X$, $X_u \cap X_v = \{w\}$, $\partial(X_u) \subseteq \{u,w\}$, and $\partial(X_v) \subseteq \{v,w\}$. Then $\ecc(u,v,G)$ equals the maximal elements of the set
\begin{eqnarray*}
  \{(d_u,\min\{d_u+1,d_w+1\}) \ |\ (d_u,d_w)\in \ecc(u,w,X_u)\} & \cup \\
  \{(\min\{d_w+1,d_v+1\},d_v) \ |\ (d_w,d_v)\in\ecc(w,v,X_v)\}.&
\end{eqnarray*}
\end{lemma}
\begin{proof}
It is clear from the fact that a shortest path from $X_u \setminus \{u\}$ to $u$ does not go
through $X_v \setminus \{w\}$ (as it should go through $w\in N(u)$), from the fact
that a shortest path from $X_u$ to $v$ goes through $\{u,w\}\subseteq
N(v)$, and from the fact that any subpath of a shortest path is a
shortest path (for some pair of vertices).
\end{proof}

Given a connected outerplanar graph $G$, for any two vertices $u,v\in V(G)$ and any vertex set $X\subseteq V(G)$ with $u,v \in X$ such that $\partial(X) \subseteq \{u,v\}$, let us define $\bestECC_D(u,v,X)$ as the minimal elements of the set
\[
  \left\{
    \ecc(u,v,H) \ \left|
    \parbox{8.8cm}{$H$ is a diameter-$D$ outerplanar completion of $G[X]$  such that $uv \in E(H)$ and such that $uv$ lies on the outer face.}\right.
  \right\}
\]
If this set is empty, we set $\bestECC_D(u,v,X)$ to $(\infty,\infty)$. Similarly to Subsection~\ref{subsec:cutvertices}, we may drop the subindex $D$ from the notation $\bestECC_D(u,v,X)$. Here, $\ecc(u,v,H) \le \ecc(u,v,H')$ if and only if for any $(d_1,d_2)\in\ecc(u,v,H)$ there exists a $(d'_1,d'_2)\in \ecc(u,v,H')$ such that $(d_1,d_2) \le (d'_1,d'_2)$. According to the possible forms of $\ecc(u,v,H)$, we have that $\bestECC(u,v,X)$ is of one of the following five forms:
\begin{itemize}
  \item[$\bullet$] $\{\{(d,d)\}\}$,
  \item[$\bullet$] $\{\{(d,d+1)\}\}$,
  \item[$\bullet$] $\{\{(d+1,d)\}\}$,
  \item[$\bullet$] $\{\{(d,d+1),(d+1,d)\}\}$, or
  \item[$\bullet$] $\{\{(d,d+1)\},\{(d+1,d)\}\}$,
\end{itemize}
for some positive integer $d$.

Considering $\bestECC(u,X)$ for some~$u$ and~$X$, note that $u$ has at least one incident edge $uv$ on the outer face in an outerplanar completion achieving $\bestECC(u,X)$. Thus, we can observe the following.

\begin{obs}\label{obs:ecc-u-ecc-uv}
  $\bestECC(u,X) = \min_{v\in X}\min_{S\in\bestECC(u,v,X)}\max_{(d_u,d_v)\in S} d_u$.
\end{obs}


\subsection{The algorithm for connected outerplanar graphs}
\label{subsec:dp-prog}

We now proceed to describe a polynomial-time algorithm that solves \OPDIfull when the input outerplanar graph is assumed to be connected. In Subsection~\ref{subsec:disconnected} we will deal with the disconnected case.
In a graph, a \emph{block} is either a 2-connected component or a bridge.
Before proceeding to the formal description of the algorithm, let us provide a high-level sketch.

Algorithm~\ref{alg:connected} described below receives a \emph{connected} outerplanar graph $G$, an arbitrary non-cut vertex $r$ of $G$, called the \emph{root} (such a vertex is easily seen to exist in any graph), and a positive integer $D$. In order to decide whether $G$ admits a diameter-$D$ outerplanar completion, we will compute in polynomial time the value of $\bestECC_D(r,G)$, which, by definition, is distinct from $\infty$ if and only if $G$ admits a diameter-$D$ outerplanar completion.

In order to compute $\bestECC_D(r,G)$, the algorithm proceeds as follows. In the first step (lines~\ref{ln:first line}-\ref{ln:remove blocks}), we consider an arbitrary block $B$ of $G$ containing $r$ (line~\ref{ln:first line}), and in order to reduce the input graph $G$, we consider all cut vertices $v$ of $G$ in $B$. For each such cut vertex $v$, we order its corresponding branches according to their eccentricity w.r.t. $v$ (line~\ref{ln:reorder blocks}), and by Corollary~\ref{cor-7-blocks} it is safe to keep just a constant number of them, namely 8 (line~\ref{ln:remove blocks}). For computing the eccentricity of the branches not containing the root (lines~\ref{ln:for i}-\ref{ln:recursion}), the algorithm calls itself recursively, by considering the branch as input graph, and vertex $v$ as the new root.

In the second step of the algorithm (lines~\ref{ln:foreach uvX}-\ref{ln:update table}),  we try all 2-vertex separators $u,v$ in the eventual completed graph $G'$ (note that $G'$ cannot be 3-connected, as otherwise it would contain a $K_{2,3}$-minor), together with a set $X$ consisting of a subset of the connected components of $G' \setminus \{u,v\}$, not containing the root $r$. For each such triple $(u,v,X)$, our objective is to compute the value of $\bestECC_D(u,v,X)$. For doing so, after initializing its value (lines~\ref{ln:init table}-\ref{ln:init table 2}), we consider all possible triples $w, X_u, X_v$  chosen as in Lemma~\ref{claim:ecc-triangle} after adding the triangle  $uvw$ to $G[X]$ (line~\ref{ln:foreach wXuXv}), for which we already know the values of $\bestECC_D(u,w,X_u)$ and $\bestECC_D(w,v,X_v)$, since the sets $X$ are processed by increasing size. Among all choices of one element in $\bestECC_D(u,w,X_u)$ and another in $\bestECC_D(w,v,X_v)$ (line~\ref{ln:foreach Su and Sv}), only those whose corresponding completion achieves diameter at most $D$ are considered for updating the value of  $\bestECC_D(u,v,X)$ (line~\ref{ln:check Su Sv}). For updating $\bestECC_D(u,v,X)$ (line~\ref{ln:update table}), we first compute $\ecc_D(u,v,X)$ using Lemma~\ref{claim:ecc-triangle} (line~\ref{ln:compute Ecc}).

Finally, once we have computed all values of $\bestECC_D(u,v,X)$, we can easily compute the value of $\bestECC_D(u,X)$ by using Observation~\ref{obs:ecc-u-ecc-uv} (line~\ref{ln:final value}). 
We can now provide a formal description of the algorithm.

\newcommand{\DPtable}{\ensuremath{\text{Tab}_{\text{ECC}}}}
\vspace{.65cm}
\begin{algorithm}[H]
  \DontPrintSemicolon
  \LinesNumbered

  \SetKwInOut{Input}{Input}
  \SetKwInOut{Output}{Output}
  \Input{A connected outerplanar graph $G$, a root $r\in V(G)$ such that $G \setminus \{r\}$ is connected, and a positive integer $D$.}
  \Output{$\bestECC_D(r,G)$.\vspace{.2cm}}
  \tcp{all over the recursive calls of the algorithm, $G$ is a global variable, which gets updated whenever some vertices are removed in line~\ref{ln:remove blocks}}

  \nllabel{ln:first line} Let $B$ be a block of $G$ containing $r$\;

  \tcp{we consider all cut vertices of $B$ and we reduce $G$}
  \ForEach{cut vertex $v \in V(B)$}{
    Let $C_0,\ldots, C_t$ be the connected components of $G\setminus \{v\}$, where $r\in C_0$\;
    Let $B_0 \gets C_0$ \hfill\rtcp{the branch containing the root}
    \For{$i\gets 1$ \KwTo $t$\nllabel{ln:for i}}{
      Let $B_i \gets G[C_i\cup\{v\}]$\nllabel{ln:get branches}\rtcp{the branches around $v$}
      $\text{ecc}_i \gets$ \texttt{OPDI-Connected}$(B_i,v,D)$.\nllabel{ln:recursion}\rtcp{recursive call to compute $\bestECC_D(v,B_i)$}
    }
    Reorder the $B_i$'s so that $\text{ecc}_1 \ge \text{ecc}_2 \ge \ldots \ge \text{ecc}_t$\nllabel{ln:reorder blocks}\;
    Remove $B_8,\ldots,B_t$ from $G$ \hfill\rtcp{by Corollary~\ref{cor-7-blocks}}\nllabel{ln:remove blocks}
  }

  \tcp{guess all 2-separators $u,v$ in the target completion $G'$, together with a subset $X$ of the connected components of $G' \setminus \{u,v\}$}

  \ForEach{triple $(u,v,X)$ such that $r\notin X\setminus \{u,v\}$ and $\partial(X) \subseteq \{u,v\}$\nllabel{ln:foreach uvX}}{
    \tcp{by increasing size of $X$, and only if the triple $(u,v,X)$ has not already been considered before in a previous iteration}
    $\DPtable(u,v,X) \gets \{\{(\infty,\infty)\}\}$. \ \rtcp{it corresponds to $\ecc_D^*(u,v,X)$}\nllabel{ln:init table}
    \lIf{$X=\{u,v\}$}{
      $\DPtable(u,v,X) \gets \{\{(0,1),(1,0)\}\}$.\;\nllabel{ln:init table 2}
    }
    \lElse{
      \ForEach{$w, X_u, X_v$ satisfying the hypothesis of Lemma~\ref{claim:ecc-triangle} in the graph obtained from $G[X]$ by adding the triangle  $uvw$
      \nllabel{ln:foreach wXuXv}}{
        \tcp{eccentricities of smaller subgraphs have been already computed}
        \ForEach{$S_u\in\emph{\DPtable}(u,w,X_u)$ and $S_v\in\emph{\DPtable}(w,v,X_v)$\nllabel{ln:foreach Su and Sv}}{
          \If{for all $(d_u,d_w^u) \in S_u$ and $(d_v,d_w^v) \in S_v$, we have $(d_w^u + d_w^v \leq D)$ or $(d_u + 1 + d_v \leq D)$\nllabel{ln:check Su Sv}}{
            \tcp{if the diameter of the considered completion of $X$ is $\leq D$, we compute $\ecc_D(u,v,X)$ using Lemma~\ref{claim:ecc-triangle}}

            $\text{Ecc} \gets \max\left\{\protect\parbox{7.5cm}{
              $\{(d_u,\min\{d_u+1,d_w+1\}) \ |\ (d_u,d_w)\in S_u \} \cup \{(\min\{d_v+1,d_w+1\},d_v)\ |\ (d_w,d_v)\in S_v\}\}$
            }\right\}$\nllabel{ln:compute Ecc}\;
          \tcp{update $\ecc_D^*(u,v,X)$}
          $\DPtable(u,v,X) \gets \min\{\DPtable(u,v,X) \cup \text{Ecc}\}$\nllabel{ln:update table}\;
          }
          }
          }
          }
          }
  \tcp{finally, we compute $\bestECC_D(r,G)$ using  Observation~\ref{obs:ecc-u-ecc-uv}}
  \Return $\min_{v\in V(G)\setminus \{r\}} \min_{S\in\text{Tab}_{\text{ECC}}(r,v,V(G))}\max_{(d_u,d_v)\in S} d_u$.\nllabel{ln:final value}
  \caption{\texttt{OPDI-Connected}\label{alg:connected}}
\end{algorithm} \vspace{0.3cm}

The correctness of Algorithm~\ref{alg:connected} follows from the results proved in Subsections~\ref{subsec:cutvertices} and~\ref{subsec:2vertexSep}, and the following fact (whose proof is straightforward), which guarantees that the value of $\ecc_D^*(u,v,X)$ can indeed be computed as done in lines~\ref{ln:foreach wXuXv}-\ref{ln:update table}.

\begin{fact}\label{fact:ejk}
There exists an outerplanar completion $H$ of $G[X]$ with the edge $uv$ on the outerboundary if and only if there is $w\in X$ and two sets $X_u, X_v$ such that:
\begin{itemize}
\item[(a)] $X_u\cup X_v=X$, $X_u\cap X_v=\{w\}$,\vspace{-.1cm}
\item[(b)]$\partial_G(X_u)\subseteq \{u,w\}$ and $\partial_G(X_v)\subseteq \{v,w\}$, and \vspace{-.1cm}
\item[(c)] there exists an outerplanar completion $H_u$ of $G[X_u]$ with the edge $uw$ on the outerboundary, and an outerplanar completion $H_v$ of $G[X_v]$ with the edge $vw$ on the outerboundary.
\end{itemize}
\end{fact}

It remains to analyze the running time of the algorithm.

\paragraph{\textbf{Running time analysis of Algorithm~\ref{alg:connected}}.} Note that at line~\ref{ln:get branches} each $B_i$ is recursively replaced by an equivalent (by Corollary~\ref{cor-7-blocks}) subgraph such that its cut vertices have at most 8 branches attached.


Let us first focus on the second step of the algorithm, that is, on lines~\ref{ln:foreach uvX}-\ref{ln:update table}. The algorithm considers in line~\ref{ln:foreach uvX} at most $O(n^2)$ pairs $\{u,v\}$. As each of $u$ and $v$ has at most 7 attached branches avoiding the root, and $G \setminus \{u,v\}$ has at most 2 connected components with vertices adjacent to both $u$ and $v$ (as otherwise $G$ would contain a $K_{2,3}$-minor), there are at most $2^{7} \cdot 2^{7} \cdot 2^{2} = 2^{16}$ possible choices for assigning these branches or components to $X$ or not. In line~\ref{ln:foreach wXuXv}, the algorithm considers $O(n)$ vertices $w$. Similarly, as $w$ belongs to at most 7 branches not containing $u$ nor $v$, there are at most $2^7$ choices for assigning these branches to $X_u$ or $X_v$. In lines~\ref{ln:foreach Su and Sv}-\ref{ln:update table}, the algorithm uses values that have been already computed in previous iterations, as the sets $X$ are considered by increasing order. Note that each of $\bestECC_D(u,w,X_u)$ and $\bestECC_D(w,v,X_v)$ contains at most 2 elements, so at most 4 choices are considered in line~\ref{ln:foreach Su and Sv}. Again, at most 4 choices are considered in line~\ref{ln:check Su Sv}. Therefore, lines~\ref{ln:foreach Su and Sv}-\ref{ln:update table} are executed in constant time.

As for the first step of the algorithm (lines~\ref{ln:first line}-\ref{ln:remove blocks}), the algorithm calls itself recursively. The number of recursive calls is bounded by the number of blocks of $G$, as by construction of the algorithm each block is assigned a single root. Therefore, the number of recursive calls is $O(n)$. Once the algorithm calls itself and the corresponding branch has no cut vertex other than the root, the algorithm enters in lines~\ref{ln:foreach uvX}-\ref{ln:update table}, whose time complexity has already been accounted above. (Note that each triple $(u,v,X)$ is considered only once, and the value of $\ecc^*_D(u,v,X)$ is stored in the tables.)

Finally, in line~\ref{ln:final value}, the algorithm considers $O(n)$ vertices, and for each of them it chooses among constantly many numbers.  Summarizing, we have that the algorithm has overall complexity $O(n^3)$. 

It is worth mentioning that Algorithm~\ref{alg:connected} can also compute the actual completion achieving diameter at most $D$, if any, within the same time bound. Indeed, it suffices to keep track of which edges have been added to $G$ when considering the guessed triangles $uvw$ (recall that we may assume that the completed graph is triangulated).


\begin{theorem}\label{thm-algo-connected}
Algorithm~\ref{alg:connected} solves \OPDIfull for connected input graphs in time $O(n^3)$.
\end{theorem}


Note that we can compute ${\bf opdi}(G)$ by calling Algorithm~\ref{alg:connected} with an arbitrary root~$r\in V(G)$ (such that $G \setminus \{r\}$ is connected) for increasing values of~$D$.

%
%
%

\begin{corollary}\label{cor-opt-connected}
  Let~$G$ be a connected outerplanar graph. Then, ${\bf opdi}(G)$ can be computed in time $O(n^4)$.
\end{corollary}


\subsection{The algorithm for disconnected outerplanar graphs}
\label{subsec:disconnected}
\newcommand{\minrad}[1]{\ensuremath{r^*({#1})}}
\newcommand{\minecc}[1]{\ensuremath{r^+({#1})}}

In this subsection we will focus on the case where the input outerplanar graph is disconnected. The \emph{radius} of a graph is defined as the eccentricity of a ``central'' vertex, that is, the minimum eccentricity of any of its vertices.
\begin{lemma}[\!\!\cite{journal-DEGMS14}, Theorem~3]
  Let~$G$ be a maximal outerplanar graph of diameter~$D$ and radius~$r$.
  Then, $r\leq\lfloor D/2\rfloor+1$.
  \label{lem:MOPs ecc}
\end{lemma}

In the following, we denote the minimum radius of a diameter-$D$ outerplanar
completion of a graph or connected component~$G$ by~$\minrad{G}$. If $G$ has no diameter-$D$ outerplanar
completion, then let~$\minrad{G} = \infty$.


\begin{definition}\label{def:escalate}
  Let~$G$ be a connected graph and let~$D$ be an integer.  Let~$G'$ be
  the graph resulting from~$G$ by adding an isolated vertex~$v$.
  Let~$G^*$ be a diameter-$D$ outerplanar completion of~$G'$ that
  minimizes the eccentricity of~$v$.  Then, $G^*$ is called
  \emph{escalated completion} of~$(G,D)$ with respect to~$v$ and the
  eccentricity~$\ecc(v,G^*)$, denoted by~$\minecc{G}$, is called
  \emph{escalated eccentricity} of~$(G,D)$.  Again, if such a $G^*$
  does not exist, let $\minecc{G} = \infty$.
\end{definition}

\noindent We will apply Definition~\ref{def:escalate} also to
connected components of a graph and, if clear from context, we
omit~$D$. Note that we can compute~$\minecc{G}$ by guessing an edge
between the isolated vertex~$v$ and~$G$ and running
$\texttt{OPDI-Connected}$, the algorithm for connected
graphs. Hence this can be done in $O(n^4)$~time.  Also note
that~$\minrad{G}\leq\minecc{G}\leq\minrad{G}+1$.

\begin{lemma}\label{lem:dis_small_ecc}
Given a graph $G$ with a connected component $C$ such that $\minecc{C}
< D/2$, then $G$ has a diameter-$D$ outerplanar completion if and only
if $G\setminus C$ does.
\end{lemma}
\begin{proof}
In a diameter-$D$ outerplanar completion of $G\setminus C$ there is a
vertex $v$ with eccentricity at most $\lfloor D/2\rfloor+1$, by
Lemma~\ref{lem:MOPs ecc}.  In this completion, adding the completion
of $C+v$ achieving $\minecc{C} < D/2$, yields a diameter-$D$
outerplanar completion of $G$.
\end{proof}

\newcommand{\C}{\ensuremath{\mathcal{C}}}
\newcommand{\Cmax}{\ensuremath{\C_{\text{max}}}}
\newcommand{\Ctwo}{\ensuremath{\C_2}}

\begin{obs}\label{obs:escalated cutset}
  Let~$C$ be a connected component of~$G$, let~$G'$ be an outerplanar
  completion of~$G$ and let~$C'$ be a connected component of~$G'
  \setminus C$.
  Then, there is a vertex~$v\in C$ at distance at least~$\minecc{C}$ to each vertex of~$C'$ in~$G'$.
\end{obs}
\begin{proof}
  Let the result of contracting all vertices in~$G' \setminus (C\cup
  C')$ onto vertices in~$C$ and contracting~$C'$ onto a single
  vertex~$u$ be~$G''$.  Then, $G''$ is a subgraph of an outerplanar
  completion of the result of adding~$u$ as isolated vertex
  to~$G'[C]$.  By definition, $\ecc(u,G'')\geq\minecc{C}$, implying
  that there is a vertex~$v\in C$ at distance at least~$\minecc{C}$
  to~$u$ in~$G''$.  Thus, $v$ is at distance at least~$\minecc{C}$ to
  each vertex of~$C'$ in~$G'$.
\end{proof}

Observation~\ref{obs:escalated cutset} immediately implies that any
cutset separating two connected components~$C_1$ and~$C_2$ of~$G$
in~$G'$ has distance at least~$\minecc{C_1}$ and~$\minecc{C_2}$ to
some vertex in~$C_1$ and~$C_2$, respectively. Thus, these two vertices
are at distance at least~$\minecc{C_1}+\minecc{C_2}$ in~$G'$.
\begin{corollary}\label{cor:connect big}
  Let $C_1$ and~$C_2$ be connected components of~$G$ such
  that~$\minecc{C_1}+\minecc{C_2}>D$ and let~$G'$ be a diameter-$D$
  outerplanar completion of~$G$.  Then, $C_1$ and~$C_2$ are adjacent
  in~$G'$, i.e. $G'$ has an edge with an end in $C_1$ and an end in
  $C_2$.
\end{corollary}

Corollary~\ref{cor:connect big} allows us to conclude that all
connected components~$C$ with~$\minecc{C}> D/2$ have to be pairwise
adjacent in any diameter-$D$ outerplanar completion of~$G$. Thus,
there cannot be more than three such components.
\begin{lemma}\label{lem:dis_few_big_ecc}
An outerplanar graph $G$ with more than 3 connected components $C$
such that $\minecc{C}> D/2$ has no diameter-$D$ outerplanar
completion.
\\
On the other hand, if $G$ has no connected component $C$
such that $\minecc{C}> D/2$, then $G$ necessarily has a diameter-$D$
outerplanar completion.
\end{lemma}
\begin{proof}
The first statement comes from the above comments. The proof of the
second statement is similar to the one of
Lemma~\ref{lem:dis_small_ecc}. For some component $C$ of $G$, let $v$
be such that $ecc(v,C)= \minrad{C} \le \minecc{C}\le D/2$, and
complete $C$ in order to achieve this value. Then for the other
components $C'$ consider their escalated completion with respect to
$v$. As $\minecc{C'}\le D/2$ this graph has diameter at most $D$.
\end{proof}

Hence, assume $G$ has $p=1,2$, or $3$ connected components $C$ such
that $\minecc{C}> D/2$. By Corollary~\ref{cor:connect big} these $p$
components are pairwise adjacent in the desired completion. Note that
with~$O(n^{2p-2})$ tries, we can guess $p-1$ edges connecting all such
components into one larger component.  Thus, in the following, we
assume that there is only one component~$C$ with~$\minecc{C}> D/2$,
denoted $\Cmax$.

\begin{lemma}\label{lem:dis_few_big_ecc_2}
Consider an outerplanar graph $G$ with exactly one connected
components $\Cmax$ such that $D/2 < \minecc{\Cmax} < \infty$. If
$\minrad{\Cmax}\le D/2$, then $G$ necessarily has a diameter-$D$
outerplanar completion.
\end{lemma}
\begin{proof}Same proof as Lemma~\ref{lem:dis_few_big_ecc}.\end{proof}


Let us now distinguish two cases according to the parity of $D$.

\begin{lemma}\label{lem:dis_odd_easy}
For odd $D$, if an outerplanar graph $G$ has at most one component
$\Cmax$ such that $D/2 < \minecc{\Cmax} < \infty$, then $G$ has a
diameter-$D$ outerplanar completion.
\end{lemma}
\begin{proof}
Indeed, by Lemma~\ref{lem:dis_small_ecc} it is sufficient to consider
the component $\Cmax$ alone. As $ \minecc{\Cmax} < \infty$, $\Cmax$
has a diameter-$D$ outerplanar completion, and so does $G$.
\end{proof}

The case where $D$ is even is more technical.

\begin{lemma}\label{lem:dis_even_easy_2}
For even $D$, Let $p$ and $q$ respectively denote the number of
connected components $C$ such that $D/2 < \minecc{C} < \infty$ and $\minecc{C}
= D/2$, of an outerplanar graph $G$.  If $p\ge 2$ and $p+q \ge 5$, then
$G$ has no diameter-$D$ outerplanar completion.
\end{lemma}
\begin{proof}
By Corollary~\ref{cor:connect big}, in a diameter-$D$ outerplanar
completion $G'$ of $G$ the $p$ components are pairwise adjacent, and
any of the $q$ components is adjacent to every of the $p$ ones.  For
$p=2$, as $q\ge 3$, this would induce a $K_{2,3}$-minor in $G'$, a
contradition. For the other cases, this would induce a $K_{4}$-minor
in $G'$, a contradition.
\end{proof}

\begin{lemma}\label{lem:dis_even_easy}
  For even $D$, if an outerplanar graph $G$ has one component, denoted
  $\Cmax$, such that $D/2 < \minrad{\Cmax} < \infty$, and at least 4 other
  components $C$ such that $D/2 \le \minecc{C} < \infty$, then $G$ has no
  diameter-$D$ outerplanar completion.
\end{lemma}
\begin{proof}
Let us denote $C_1,C_2,C_3,$ and $C_4$ the connected components such
that $\minecc{C_i} \ge D/2$, distinct from $\Cmax$. Assume for
contradiction that $G$ admits a diameter-$D$ outerplanar completion,
denoted $G'$.

  \begin{claimN}\label{cl:pairwise adj}
    For each~$C_i,C_j$, either $C_i$ and $C_j$ are adjacent in
    $G'$, or $C_i$ and $C_j$ have a common neighbor in $G'$.
  \end{claimN}
  \begin{proof}
    Assume for contradiction that $C_i$ and $C_j$ are not adjacent and
    do not have a common neighbor in $G'$. Let us now construct the
    graph $G''$ from $G'$ as follows. For any component $C$ of
    $G'\setminus (C_i\cup C_j)$ that is not adjacent to both $C_i$ and
    $C_j$, contract $C$ onto vertices of $C_i$ or $C_j$ (According to
    the one $C$ is neighboring). As $G''$ is obtained from $G'$ by
    contracting edges, $G''$ also is a diameter-$D$ outerplanar
    completion (for some graph containing $C_i$ and $C_j$).  Let
    $N_i:=N_{G''}(C_i)$, let~$N_j:=N_{G''}(C_j)$, and note that
    $C_i\cap N_j=\emptyset$, $N_i\cap C_j=\emptyset$, and $N_i\cap
    N_j=\emptyset$.  Then, by Observation~\ref{obs:escalated cutset}
    (as $G''\setminus C_i$ and $G''\setminus C_j$ are connected),
    there are vertices~$v_i\in C_i$ and~$v_j\in C_j$ at distance at
    least~$D/2$ to each vertex in~$N_i$ and~$N_j$, respectively,
    in~$G''$. Since~$N_i$ and~$N_j$ are at distance at least one,
    $v_i$ and~$v_j$ are at distance at least~$D+1$,
    contradicting~$G''$ having diameter~$D$.
  \end{proof}
  \begin{claimN}\label{cl:3 neighbors}
    There is a vertex~$u\in\Cmax$ that is adjacent in $G'$ to 3 of the
    components $C_1,C_2,C_3,$ and $C_4$.
  \end{claimN}
  \begin{proof}
    First, note that there is a vertex~$u$ and 3 components, say
    $C_1,C_2,C_3$, with~$u\in N_{G'}[C_i]$ for all~$1\le i\le 3$,
    since otherwise, there would be internally vertex-disjoint paths
    between each two of the four components $C_i$, implying the
    existence of a $K_4$-minor in~$G'$.

    If~$u$ is neither in~$\Cmax$ nor in $C_i$, for $1\le i\le 3$,
    then, since all the~$C_i$ are adjacent to~$\Cmax$ (by
    Corollary~\ref{cor:connect big}), $G'$ would have a
    $K_{2,3}$-minor on the vertex sets~$\{u,\Cmax\}$
    and~$\{C_1,C_2,C_3\}$.

    Hence, in the following, we assume that~$u\in C_1$.  Let~$z$ be a
    neighbor of~$C_1$ in~$\Cmax$ and, for~$i\in\{2,3\}$ let~$w_i$
    denote a neighbor of $C_4$ in $N[C_i]$. We note that $w_2\ne z$
    and $w_3\ne z$, since otherwise, the claim follows and we are
    done.  Furthermore, $w_2\ne u$ and $w_3\ne u$, since otherwise
    there is a $K_{2,3}$-minor on the vertex sets~$\{u,\Cmax\}$
    and~$\{C_2,C_3,C_4\}$.
%
    Let~$X:=(C_4\cup\{w_2,w_3\})\setminus(C_2\cup C_3)$ and note that
    $X$ is adjacent to~$C_2$ and $C_3$, respectively.  Let~$Y$ be the
    connected component of $\Cmax\setminus\{w_2,w_3\}$ containing $z$,
    and note that $Y$ is adjacent to $C_1$ and $X$. Finally, since
    $X$, $Y$, $C_1$, $C_2$, and $C_3$ are pairwise disjoint, $G'$ has
    a $K_{2,3}$-minor on the vertex sets~$\{X,C_1\}$
    and~$\{C_2,C_3,Y\}$.
  \end{proof}

  Let~$v$ denote a vertex of~$\Cmax$ that is at distance at
  least~$D/2+1$ to~$u$ in~$G'$ and consider the result~$G' \setminus
  \{u\}$ of removing~$u$ from~$G'$. Let~$C$ denote the connected
  component of~$G' \setminus \{u\}$ that contains~$v$.  Towards a
  contradiction, assume there is a connected component~$C_i$ that is
  adjacent to~$u$ but not to~$C$ in~$G'$, then all paths between~$v$
  and any vertex in~$C_i$ contain~$u$. Since~$G'$ has diameter~$D$,
  all vertices in~$C_i$ are at distance at most~$D/2-1$ to~$u$
  in~$G'$, contradicting~$\minecc{C_i}\ge D/2$.  Thus there is a
  $K_{2,3}$-minor in~$G'$ on the vertex sets~$\{C_1,C_2,C_3\}$
  and~$\{u,X\}$ where $X$ is the connected component of $G'\setminus
  (C_1\cup C_2\cup C_3\cup \{u\})$ containing $v$. This concludes the
  proof of the lemma.
\end{proof}

Hence, assume $G$ has $q=0,1,2$, or $3$ connected components $C$ such
that $\minecc{C} = D/2$. By Corollary~\ref{cor:connect big} these $q$
components are adjacent to each of the $p$ components such that
$\minecc{C} > D/2$.  Note that with~$O(n^{2q})$ tries, we can guess
$q$ edges connecting each of the $q$ components to one of the $p$
component. Then we are left with a connected graph, and we can call
$\texttt{OPDI-Connected}$.
\newpage

\noindent{\bf The algorithm itself.} We now describe a polynomial-time
algorithm that solves the \OPDIfull problem when the input contains a
disconnected outerplanar graph. Algorithm~\ref{alg:disconnectedOPT}
described below receives a (disconnected) outerplanar graph $G$, and a
positive integer $D$.

\vspace{.3cm}
\begin{algorithm}[H]
  \DontPrintSemicolon
  \LinesNumbered
  \SetKwInOut{Input}{Input}
  \SetKwInOut{Output}{Output}
  \Input{A disconnected outerplanar graph $G=(V,E)$ and an integer~$D$.}
  \Output{`\TRUE'' if and only if~$G$ has a diameter-$D$ outerplanar completion.\vspace{.2cm}}

  \ForEach{connected component~$C$ of~$G$}{
    $\minecc{C}\gets \infty$\;
    $\minrad{C}\gets \infty$\;
    \ForEach{$u\in V(C)$}{
      Ecc$\gets \texttt{OPDI-Connected}(C,u,D)$\;
      $\minrad{C} \gets \min \{\minrad{C},{\rm Ecc}\}$\;
      $C'\gets C$ with added vertex $v$ and added edge $uv$\;
      Ecc$\gets \texttt{OPDI-Connected}(C',v,D)$\;
      $\minecc{C} \gets \min \{\minecc{C},{\rm Ecc}\}$\;
    }
    \lIf{$\minecc{C} = \infty$}{\Return{\FALSE}\;}
    \lIf{$\minecc{C} < D/2$}{Remove $C$ from $G$\;}
  }
  \lIf{$\minecc{C}\le D/2$ for every $C$}{\Return{\TRUE}\;}
  \If{$\minecc{C}\le D/2$ for every $C$ except one, $\Cmax$, and $\minrad{\Cmax}\le D/2$}{\Return{\TRUE}\;}

  \lIf{$\minecc{C} > D/2$ for at least 4 conn. components $C$}{
    \Return{\FALSE}\;
  }
  \ForEach{choice of edges interconnecting these $p=1,2$, or $3$ connected components\nllabel{ln:guess connections}}{
    \tcp{choose $p-1$ edges}
    Let $\Cmax$ be this new conn. component\;
    \If{$\texttt{\emph{OPDI-Connected}}(\Cmax,v,D) < \infty$\nllabel{ln:call connected Cmax}}{
      \lIf{$D$ is odd}{\Return{\TRUE}}
      \If{$G$ has more than $5-p$ conn. comp.}{
        \tcp{$\Cmax$ and $q$ connected comp. such that $\minecc{C} = D/2$}
        \Return{\FALSE}
      }\Else{
        \ForEach{choice of $q$ edges connecting $G$\nllabel{ln:guess q edges}}{
          \lIf{$\texttt{\emph{OPDI-Connected}}(G,v,D) < \infty$\nllabel{ln:call connected G}}{\Return{\TRUE}}
        }
        \Return{\FALSE}\;
      }
    }
  }
  \Return{\FALSE}

\caption{\texttt{OPDI-Disconnected}\label{alg:disconnectedOPT}}
\end{algorithm}
\vspace{.3cm}

At the beginning, the algorithm computes
$\minecc{C}$ and $\minrad{C}$ for each connected component~$C$
of~$G$. For computing $\minecc{C}$ the algorithm adds a vertex $v$,
guessing (with $O(n)$ tries) an edge connecting $v$ to $C$, and then
calls $\texttt{OPDI-Connected}$ for this component and root $v$.  For
computing $\minrad{C}$ the algorithm guesses a root $u$ (with $O(n)$
tries), and then calls $\texttt{OPDI-Connected}$ for $C$ and root $u$.

If $\minrad{C}=\infty$ for some component $C$ then, as $\minrad{G} \ge
\minrad{C}$, $G$ has no diameter-$D$ outerplanar completion.

Then, as they could be added in a diameter-$D$ outerplanar completion
(by Lemma~\ref{lem:dis_small_ecc}), the algorithm removes the
components $C$ with small escalated eccentricity, that is those such
that $\minecc{C} < D/2$.

Then the algorithm tests if there is no component $C$ such that
$\minecc{C} > D/2$, or if there is only one component $C$ such that
$\minecc{C} > D/2$, and if $\minrad{C} \le D/2$.  In both cases by
Lemma~\ref{lem:dis_few_big_ecc} and Lemma~\ref{lem:dis_few_big_ecc_2},
$G$ is a positive instance.

Then the algorithm tests if there are more than 3 components $C$ such
that $\minecc{C} > D/2$.  In this case, by
Lemma~\ref{lem:dis_few_big_ecc}, $G$ is a negative instance.
Otherwise, $G$ has $p=1$, $2$, or 3 such connected components, and the
algorithm guesses $p-1$ edges (in time $O(n^{2p-2})$) to connect them
(as they should be by Corollary~\ref{cor:connect big}). For each such
graph we call algorithm $\texttt{OPDI-Connected}$ to check that this graph
has a diameter-$D$ outerplanar completion.

Then the algorithm proceeds differently according to $D$'s parity.  If
$D$ is odd, then $G$ is a positive instance (By
Lemma~\ref{lem:dis_odd_easy}). If $D$ is even, if $G$ has (still) more
than $5-p$ connected components (by Lemma~\ref{lem:dis_even_easy_2}
and Lemma~\ref{lem:dis_even_easy}), then $G$ is a negative instance.
Then we are left with a graph $G$ with $1+q$ connected components, and
again the algorithm guesses $q$ edges (in time $O(n^{2q})$),
connecting $G$. For each of these graphs the algorithm calls
$\texttt{OPDI-Connected}(G,v,D)$ (for any $v$) to check whether this
graph admits a diameter-$D$ outerplanar completion.

Finally if none of these ``guessed'' connected graphs has a
diameter-$D$ outerplanar completion, then the algorithm concludes that
$G$ is a negative instance.

\begin{theorem}\label{thm:dis_comp}
Algorithm~\ref{alg:disconnectedOPT} solves \OPDIfull for disconnected
input graphs in polynomial time.  For odd $D$ the running time is
$O(n^7)$, while it is $O(n^{9})$ for even $D$.
\end{theorem}
\begin{proof}
  Indeed, the algorithm runs in time $O(n^7)$ for odd $D$ (at most $O(n^4)$ at line~\ref{ln:guess connections}, times $O(n^3)$ for the call to $\texttt{OPDI-Connected}$ in line~\ref{ln:call connected Cmax}).
  The algorithm runs in $O(n^{2p+2q+1})$~time for even~$D$ ($O(n^{2p-2})$ in line~\ref{ln:guess connections}, times $O(n^{2q})$ in line~\ref{ln:guess q edges}, times $O(n^3)$ for the call to $\texttt{OPDI-Connected}$ in line~\ref{ln:call connected G}), where $p$ and $q$ respectively denote the number of connected components $C$ such that $\minecc{C} > D/2$ and $\minecc{C} = D/2$. As $p+q \le 4$, we are done.
\end{proof}




\begin{appendix}

\section{A partial list of forbidden substructures}
\label{sec:maxParallelMatching}


In this section we provide a partial list of forbidden substructures that cannot appear as minors in graphs $G$ such that ${\bf opdi}(G) \leq D$, for any integer positive $D$. We hope that these constructions will turn useful towards obtaining a complete (finite) obstruction set. Our main tool is the notion of parallel matching, which is strongly related to the diameter of a maximal outerplanar graph (Subsection~\ref{subsec:maxParallelMatching}). We then present the forbidden substructures in Subsection~\ref{subsec:forbidden}.

\subsection{Maximum parallel matching and diameter}
\label{subsec:maxParallelMatching}

We start by defining the notion of parallel matching.\vspace{.2cm}

\begin{hey}{Parallel matching} The diameter of an outerplanar graph $G$ can be visualized by a specific kind of matching. Indeed, by defining an embedding of an outerplanar graph as a cyclic ordering of its vertex set around the outerboundary, we say that a matching $u_1v_1,\dots,u_kv_k$ is a collection of {\em parallel edges} (or a \emph{parallel matching}) if $u_1,\dots,u_k,v_k,\dots,v_1$ is cyclically ordered with respect to the embedding. We further require that some vertex should exist between $u_1$ and $v_1$, as well as between $u_k$ and $v_k$.

  \begin{center}
    \begin{tikzpicture}
      \draw (0,0) circle (1cm);
      \foreach \i in {0,...,9}{
        \draw node[circle,fill,scale=.3] (v\i) at (\i*360/10:1cm) {};
        \draw node at (\i*360/10:1.4cm) {$v_{\i}$};
      }
      \foreach \i/\j in {0/2,3/8,7/5}{
        \draw (v\i) -- (v\j);
      }

    \end{tikzpicture}
  \end{center}

Consequently, a sequence $u_1v_1,\dots,u_kv_k\in E(G)$ of parallel edges can be seen as a laminar system of 2-cuts, i.e., in every $G\backslash \{u_i,v_i\}$ the sets $\{u_1v_1,\dots,u_{i-1},v_{i-1}\}$ and $\{u_{i+1},v_{i+1},\dots,u_kv_k\}$ belong to two distinct connected components\footnote{Note that this can also be taken as a definition of parallel matchings independent of the embedding, in the specific case where $G$ is maximal.}.
In particular, each pair of vertices --including the first and the last one-- must be a 2-cut of the graph.

  \begin{center}
    \begin{tikzpicture}
      \begin{scope}[xshift=-5cm]
        \foreach \i in {-3,...,3}{
          \draw[dashed,gray,thick] (\i*.5,.5) -- +(0,-1);
        }
        \draw[thick] (-1.5,.5) -- (1.5,.5);
        \draw[thick] (-1.5,-.5) -- (1.5,-.5);
        \draw[thick] (-1.5,.5) .. controls +(-.7,0) and +(-.7,0) .. +(0,-1);
        \draw[thick] (1.5,.5) .. controls +(.7,0) and +(.7,0) .. +(0,-1);
        \draw node at (0,1) {$\geq k$ parallel edges};
        \draw node at (0,-1) {diameter $\geq k+1$};
      \end{scope}
      \begin{scope}
        \foreach \i in {-1,...,1}{
          \draw[dashed,gray,thick] (\i*1-.4,.5) .. controls +(.3,-.3) and +(-.3,-.3) .. (\i*1+.4,.5);
          \draw[dashed,gray,thick] (\i*1-.4,-.5) .. controls +(.3,.3) and +(-.3,.3) .. (\i*1+.4,-.5);
          \draw[gray,very thin] (-2.02,0) .. controls +(.2,0) and +(0,-.6) .. (\i*1+.4,.5);
          \draw[gray,very thin] (-2.02,0) .. controls +(.2,0) and +(0,-.6) .. (\i*1-.4,.5);
          \draw[gray,very thin] (-2.02,0) .. controls +(.2,0) and +(0,.6) .. (\i*1+.4,-.5);
          \draw[gray,very thin] (-2.02,0) .. controls +(.2,0) and +(0,.6) .. (\i*1-.4,-.5);

        }
        \draw[thick] (-1.5,.5) -- (1.5,.5);
        \draw[thick] (-1.5,-.5) -- (1.5,-.5);
        \draw[thick] (-1.5,.5) .. controls +(-.7,0) and +(-.7,0) .. +(0,-1);
        \draw[thick] (1.5,.5) .. controls +(.7,0) and +(.7,0) .. +(0,-1);
        \draw node at (0,1) {$\geq k$ non-parallel edges};
        \draw node at (0,-1) {proves nothing};
      \end{scope}
    \end{tikzpicture}
  \end{center}
\end{hey}

\begin{hey}{The maximum parallel matching is smaller than the diameter}
  Because each pair of vertices in a parallel matching of size $k$ is a 2-cut of the outerplanar graph and because they form a laminar collection of cuts, we can ensure that at least two vertices are at distance at least $k+1$. Indeed, the first edge $u_1v_1$ of the parallel matching cuts the graph into at least two connected components, one of which does not contain any edge of the matching. In such a connected component we take an arbitrary vertex $s$, and obtain another vertex $t$ with the same procedure applied to edge $u_kv_k$. The distance $\dist{s,t}$ is necessarily larger than $k+1$.

  \begin{center}
    \begin{tikzpicture}
      \begin{scope}[xshift=-5cm]
        \foreach \i in {-3,...,3}{
          \draw[dashed,gray,thick] (\i*.5,.5) -- +(0,-1);
        }
        \draw[thick] (-1.5,.5) -- (1.5,.5);
        \draw[thick] (-1.5,-.5) -- (1.5,-.5);
        \draw[thick] (-1.5,.5) .. controls +(-.7,0) and +(-.7,0) .. +(0,-1);
        \draw[thick] (1.5,.5) .. controls +(.7,0) and +(.7,0) .. +(0,-1);
        \draw node[circle,fill,scale=.5,label=180:$s$] at (-2.02,0) {};
        \draw node[circle,fill,scale=.5,label=0:$t$] at (2.02,0) {};
        \draw (-2.02,0) -- (-1.5,.5) -- ++(.5,-1) -- ++(.5,0) -- ++(.5,1) -- ++(.5,-1) -- ++(.5,1) -- ++(.5,0) -- ++(.52,-.5);
      \end{scope}
    \end{tikzpicture}
  \end{center}

  \begin{obs}\label{obs:M at most d}
    Let~$G$ be an outerplanar graph of diameter~$D$ and let~$M$ be a parallel matching in~$G$.
    Then, $|M|<D$.
  \end{obs}
\end{hey}

\begin{hey}{They are not equal in general}
  Of course, there may very well exist outerplanar graphs with a very small maximum parallel matching, and an arbitrary large diameter. Cycles and paths are the easiest such examples, as their natural embedding contain no non-empty parallel matching, despite having unbounded diameter.
\end{hey}

\begin{hey}{They are equivalent for maximal outerplanar graphs}
  In the specific case of maximal outerplanar graphs (whose embedding is unique), the size of the maximum parallel matching and the diameter differ by exactly one. We need the following lemma.

  \begin{lemma}
    Let $G$ be a chordal graph and $s,t\in V(G)$. There exist $\dist{s,t}-1$ vertex-disjoint cliques, each of which is an $s,t$-cut in $G$.
  \end{lemma}
  \begin{proof}
    First, let us notice that the lemma holds for $G$ if and only if it holds for $G\backslash v$, where $v\in V(G)\backslash\{ s,t\}$ is a simplicial vertex. As a result, we can assume that either $s$ or $t$ is simplicial, and by symmetry that $s$ is simplicial.

    We can now consider the graph $G'$ obtained from $G$ by contracting the closed neighborhood $N_G[s]$ onto a vertex $s'$. Since chordal graphs are closed under edge contraction, $G'$ is chordal. Thus, we can apply the lemma recursively on the pair $\{s',t\}$ of vertices with distance $\dist[G']{s',t}=\dist[G]{s',t}-1$. This yields a collection of $\dist{s,t}-2$ cliques with --together with the clique $N_G(s)$-- prove the lemma.
  \end{proof}

In the particular case of a 2-connected outerplanar graph, these cliques are all of cardinality 2. Indeed, they are all obtained from the neighborhood of a simplicial vertex, and a simplicial vertex of an outerplanar graph has maximum degree 2. Furthermore, a laminar sequence of 2-cuts in a maximal outerplanar graph is necessarily a parallel matching, as noted above.

  \begin{corollary}\label{obs:M at least d}
    Let~$G$ be a maximal outerplanar graph of diameter~$D$ and let~$M$ be a maximum parallel matching in~$G$.
    Then, $|M|=D-1$.
  \end{corollary}

%
%
%
\end{hey}

\begin{hey}{Edge-completion and maximum parallel matching}
  An embedded outerplanar graph $G$ with a maximum parallel matching of cardinality $k-1$ can be completed by addition of edges into an outerplanar graph of diameter at most $2k$. The procedure is quite simple: for each face of the outerplanar graph, add an edge from one vertex to all others and call the result~$G'$. We will prove the bound on the diameter  by showing that the size of a maximum parallel matching does not increase too much when the ``star'' edges are added.

Let us consider a maximum parallel matching of cardinality $k'$ in $G'$ (see the red edges in the picture below). We want to prove that,  between any two consecutive edges of the matching of $G'$,  there exists an edge from $G$ which can replace them both in a parallel matching. We distinguish two cases:

\begin{itemize}
\item[$\bullet$] If one of the two consecutive (red) edges is also an edge of $G$, there is nothing to prove.
\item[$\bullet$] If both of them are star edges, they belong to different faces of $G$ (as they would otherwise share an endpoint). As a result, there is an edge of $G$ which separates those two faces, and this is the edge we are looking for.
\end{itemize}
Note that if $e_1,e_2,e_3$ are three consecutive (red) edges, the edge of $G$ associated with the pair $e_1,e_2$ may be the same as the edge associated with the pair $e_2,e_3$: this happens when $e_2$ itself is an edge of $G$. This need not trouble us.\\


\begin{center}
  \begin{tikzpicture}
        \draw[green,thick] (-1.5,-.5) -- ++ (.5,1) -- ++ (.25,0) -- ++(0,-1) -- ++(.25,0) -- ++(.5,1) -- ++(.5,-1) -- ++(.55,0) -- ++(0,1);
        \foreach \i in {-3,...,3}{
          \draw[red,thick] (\i*.5,.5) -- +(0,-1);
        }
        \draw[thick] (-1.5,.5) -- (1.5,.5);
        \draw[thick] (-1.5,-.5) -- (1.5,-.5);
        \draw[thick] (-1.5,.5) .. controls +(-.7,0) and +(-.7,0) .. +(0,-1);
        \draw[thick] (1.5,.5) .. controls +(.7,0) and +(.7,0) .. +(0,-1);
        \draw[red,thick] (-2.05,-1) -- +(.5,0) node[label=0:{\parbox[t]{3.6cm}{\color{black} Parallel Matching in $G'$}}] {};
        \draw[green,thick] (-2.05,-1.5) -- +(.5,0) node[label=0:{\parbox[t]{3.5cm}{\color{black} Edges from $G$}}] {};
        \begin{scope}[xshift=4cm,scale=.8]
          \draw[red,thick] (-.5,-.5) -- (-.5,.5);
          \draw[red,thick] (.5,-.5) -- (.5,.5);
          \draw[green,thick] (-.45,.5) -- +(0,-1);
          \draw[thick,dashed] (-1.2,.5) -- (1.2,.5);
          \draw[thick,dashed] (-1.2,-.5) -- (1.2,-.5);
          \draw[very thick] (-.5,.5) -- (.5,.5);
          \draw[very thick] (-.5,-.5) -- (.5,-.5);
        \end{scope}
        \begin{scope}[xshift=4cm,yshift=-1.23cm,scale=.8]
          \draw[red,thick] (-.5,-.5) -- (-.5,.5);
          \draw[red,thick] (.5,-.5) -- (.5,.5);
          \draw[green,thick] (.5,-.5) -- (0,.5);
          \draw[thick,dashed] (-1.2,.5) -- (1.2,.5);
          \draw[thick,dashed] (-1.2,-.5) -- (1.2,-.5);
          \draw[very thick] (-.5,.5) -- (.5,.5);
          \draw[very thick] (-.5,-.5) -- (.5,-.5);
        \end{scope}
        \begin{scope}[xshift=6.3cm,yshift=-1.23cm,scale=.8]
          \draw[red,thick] (-.5,-.5) -- (-.5,.5);
          \draw[red,thick] (.5,-.5) -- (.5,.5);
          \draw[green,thick] (0,-.5) -- (0,.5);
          \draw[thick,dashed] (-1.2,.5) -- (1.2,.5);
          \draw[thick,dashed] (-1.2,-.5) -- (1.2,-.5);
          \draw[very thick] (-.5,.5) -- (.5,.5);
          \draw[very thick] (-.5,-.5) -- (.5,-.5);
        \end{scope}
        \begin{scope}[xshift=6.3cm,scale=.8]
          \draw[red,thick] (-.5,-.5) -- (-.5,.5);
          \draw[red,thick] (.5,-.5) -- (.5,.5);
          \draw[green,thick] (-.5,-.5) -- (.5,.5);
          \draw[thick,dashed] (-1.2,.5) -- (1.2,.5);
          \draw[thick,dashed] (-1.2,-.5) -- (1.2,-.5);
          \draw[very thick] (-.5,.5) -- (.5,.5);
          \draw[very thick] (-.5,-.5) -- (.5,-.5);
        \end{scope}
  \end{tikzpicture}
\end{center}

As the maximum parallel matching in $G'$ has cardinality $k'$, there are $k'-1$ different pairs of consecutive (red) edges, each associated with a (green) edge from $G$. These (green) edges do not necessarily define a parallel matching by themselves, for they may touch at their endpoints or even be equal. Hence, we can obtain a matching in $G$ by using only every second (green) edge associated with consecutive (red) edges, thus producing a parallel matching of length $\lceil \frac {k'-1} {2}\rceil$.\\[2mm]
Consequently, the completion $G'$ of $G$ has diameter $k'+1$, and the diameter of $G$ itself is at least $\lceil \frac {k'-1} {2}\rceil +1$. As
$$k'+1\leq 2\Big( \Big\lceil \frac {k'-1} {2}\Big\rceil +1\Big),$$
the completion of $G$ with stars is a 2-approximation of the completion of $G$ achieving the smallest diameter.


\end{hey}


\subsection{Forbidden substructures}
\label{subsec:forbidden}

In this subsection we present graphs $G$ with ${\bf opdi}(G) > D$. Because a supergraph $H$ of such a graph $G$ satisfies ${\bf opdi}(H)\geq {\bf opdi}(G)>D$, the following graphs can be seen as forbidden patterns for graphs which admit a diameter-$D$  outerplanar completion and we call them \emph{diameter-$D$ obstruction}.
The only diameter-1 obstruction is $4K_1$, and indeed any outerplanar graph on at least $4$ vertices has diameter strictly greater than $1$. We proceed to give a detailed description of forbidden subgraphs for larger diameters.

\newcommand{\join}[1]{\ensuremath{\operatorname{join}(#1)}}

  If~$X$ is a family of graphs, let $iX$ denote the disjoint union of~$i$ graphs of~$X$. $\join{v,iX}$ denotes the result of taking~$i$ graphs of~$X$ and the vertex~$v$ and connecting each of the~$i$ graphs to~$v$ by a single edge.
  We construct families $A_i$ and $B_i$ of diameter-$i$ obstructions:
  \begin{align*}
    A_i:=\begin{cases}
      \{2K_1\} & \text{ if $i=0$}\\
      2B_{i-1} & \text{ if $2 \divides i$}\\
      2A_{i-1} & \text{ if $2 \not\divides i$}\\
    \end{cases}
  &&&
    B_i:=\begin{cases}
      \{K_{1,3}\}    & \text{ if $i=1$}\\
      \emptyset      & \text{ if $2 \divides i$}\\
      \join{v,3B_{i-2}} & \text{ if $2 \not\divides i$}\\
    \end{cases}
  \end{align*}
  Note that all graphs in~$B_i$ are connected.


  \begin{obs}\label{obs:contraction}
    Let~$G$ be a maximal graph and let~$G'$ be a connected induced subgraph of~$G$.
    Then, there is a triangulation of~$G'$ that is also a contraction of~$G$.
  \end{obs}

  \begin{definition}
    Let~$G$ be an outperplane graph with parallel matching~$M=\{u_1v_1,u_2v_2,\ldots,u_kv_k\}$ and let~$v$ be a vertex of~$G$ such that~$v$ is between~$u_k$ and~$v_k$ in the embedding of~$G$. Then, we call $M$ \emph{compatible with~$v$}.
  \end{definition}

  \begin{theorem}\label{lem:obstruction builder}
    Let~$D\in\mathbb{N}$ and let $G$ be an outerplanar graph that contains a graph from~$A_D \cup B_D$ as a minor.
    Then $G$ has diameter at least~$D+1$.
  \end{theorem}
  \begin{proof}
    Assume for contradiction that~$G$ has diameter at most~$D$, and also that $G$ is maximal (without loss of generality). The proof is by induction on~$D$.

    \textbf{Base}: For $D=0$, note that $A_0\cup B_0=\{2K_1\}$. Consequently, $G$ contains at least two vertices and has positive diameter. If $D=1$, we have $A_1\cup B_1=\{4K_1, K_{1,3}\}$. Consequently, $G$ contains at least four vertices and has diameter at least $2$, given that $K_4$ is not outerplanar.

    \textbf{Inductive step}: Let $D\ge 2$ and assume that  the claim holds for all graphs of diameter strictly less than~$D$.

    {\it Case 1.} $2\divides D$.
    Then, $A_D\cup B_D = 2B_{D-1}$.
    Thus, there is a partition $V_1\uplus V_2$ of~$V(G)$ such that for each~$i\in\{1,2\}$, $G[V_i]$ is connected and contains some graph of~$B_{D-1}$ as a minor.
    By Observation~\ref{obs:contraction}, there is a triangulation~$G_1$ of~$G[V_1]$ which is a contraction of~$G$.
    Let~$M_1=\{u_1v_1,u_2v_2,\ldots,u_kv_k\}$ be a maximum parallel matching in~$G_1$ and note that, by Observation~\ref{obs:M at least d}, we know that~$k\geq\diam{G_1}-1$. By induction hypothesis, $k\geq D-1$.
    Without loss of generality, let~$\ell_1$ be such that~$G_2$ is embedded entirely between~$u_\ell$ and~$u_{\ell+1}$. If~$\ell\geq k/2$, then let~$M'_1:=\{u_iv_i \mid 1\leq i \leq k/2\}$. Otherwise, let~$M'_1:=\{u_iv_i \mid k/2\leq i\leq k\}$. In both cases, $M'_1$ is compatible with each vertex in~$V_2$ and~$|M'_1|\geq k/2\geq(D-1)/2$. Since~$2\divides D$, we conclude~$|M'_1|\geq D/2$
    Likewise, there is a parallel matching~$M'_2$ in~$G_2$ that is compatible with each vertex in~$V_1$ and~$|M'_2|\geq D/2$.
    Since~$V_1$ and~$V_2$ are disjoint, $M_1\cup M_2$ is a parallel matching in~$G$ and~$|M_1\cup M_2|\geq D$. Then, the claim follows by Observation~\ref{obs:M at most d}.

    {\it Case 2.} $2 \not\divides D$.
    Then, $G$ either contains a graph in~$2A_{D-1}=4B_{D-2}$, or a graph in~$\join{v,3B_{D-2}}$ as a minor.

    \begin{claimN}
      There are pairwise disjoint vertex sets~$V_1$, $V_2$ and~$\{x,y\}$ in~$G$ such that~$\{x,y\}$ is a cut edge separating~$V_1$ from~$V_2$ in~$G$ and both~$G[V_1]$ and~$G[V_2]$ are connected and contain a graph of $B_{D-2}$ as a minor.
    \end{claimN}
    \begin{proof}
      If~$G$ contains a graph in~$4B_{D-2}$, let~$V_1\uplus V_2\uplus V_3\uplus V_4$ be a vertex partition of~$G$ such that, for all~$i$, the graph~$G[V_i]$ is connected and contains a graph in~$B_{D-2}$ as a minor. Otherwise, let~$V_1\uplus V_2\uplus V_3\uplus\{v\}$ be a vertex partitioning of~$G$ such that, for all~$i$, the graph~$G[V_i]$ contains a graph in~$B_{D-2}$ as a minor.

      Let~$G'$ denote the result of contracting  each~$V_i$ into a single vertex. Since~$G'$ is outerplane, it cannot be a~$K_4$. Thus, without loss of generality, the vertices corresponding to~$V_1$ and~$V_2$ are non-adjacent in~$G'$.
      Thus, there is no edge between~$V_1$ and~$V_2$ in~$G$ and, since~$G$ is maximal, a cut-edge~$\{x,y\}$ exists in~$G[V_3\uplus V_4]$ (or in $G[V_3\uplus\{v\}]$).
    \end{proof}

    By Observation~\ref{obs:contraction}, for each~$i\in\{1,2\}$, there is a triangulation~$G_i$ of~$G[V_i]$ that is a contraction of~$G$.
    Analogously to Case~1, there is a parallel matching~$M'$ of~$G[V_1\cup V_2]$ with~$|M'|\geq D-1$. Then, $M:=M'\cup\{\{x,y\}\}$ is a parallel matching in~$G$ with~$|M|\geq D$. Then, the claim follows by Observation~\ref{obs:M at most d}
  \end{proof}
  
  \newpage

\section{NP-hardness reduction}
\label{sec:NPh}


In this section we prove the \textsc{NP}-hardness of the following problem, which can be seen as a generalization of the \textsc{Planar Diameter Improvement} problem.

\vspace{.3cm}

\noindent\begin{boxedminipage}{1.03\textwidth}
\textsc{Weighted 2-apex Diameter Improvement}\vspace{.1cm}

\begin{tabular}{ r l }
\textbf{Input:} & A 2-apex graph $G=(V,E)$ and a weight function $w:E \to \{1,2,5\}$.\\
\textbf{Output:} & A set $E' \subseteq (V \times V) \setminus E$ such that $G' = (V, E \cup E')$ is a 2-apex graph.\\
\textbf{Objective:} & Minimize the diameter of $G'$, where $w(e)=1$ for all $e \in E'$.\\
\end{tabular}
\end{boxedminipage}\vspace{.2cm}

Recall that for an integer $p \geq 1$, a graph $G=(V,E)$ is a \emph{$p$-apex} graph if there is a set $A \subseteq V$ with $|A| \leq p$ such that $G \setminus A$ is planar. For an edge-weighted graph, the distance and the diameter are defined taking into account the weights of the edges.

\begin{theorem}\label{thm:NPh} \textsc{Weighted 2-apex Diameter Improvement} is \textsc{NP}-hard.
\end{theorem}
\begin{proof} We will prove that the decision version of \textsc{Weighted 2-apex Diameter Improvement} is \textsc{NP}-complete. It is clear that the problem belongs to \textsc{NP}, as the (weighted) diameter of a graph can be computed in polynomial time. For the \textsc{NP}-hardness, we will reduce from the \textsc{Planar 3-SAT} problem where each variable occurs at most 3 times, which is known to be \textsc{NP}-complete~\cite{GareyJ79comp}. Note that we can assume that each variable appears
at least once in a positive literal and at least once in a negative literal, as otherwise we can set its value to true or false without affecting the satisfiability of the formula. More precisely, given a formula $\phi$, we will construct an edge-weighted 2-apex graph $G$ that can be completed into a 2-apex graph $G'$ of diameter at most 4 if and only if $\phi$ is satisfiable. Let us now proceed to describe the reduction.

The formula $\phi$ naturally defines a bipartite graph $G_{\phi}$ with one vertex per variable and clause, and where a variable-vertex $x_i$ is adjacent to a clause-vertex $c_j$ if and only if the corresponding variable belongs to the corresponding clause. By hypothesis, the graph $G_{\phi}$ is planar and each variable-vertex has degree at most 3. We fix an arbitrary plane embedding of $G_{\phi}$. Towards constructing~$G$, we replace each variable-vertex $x_i$ with the gadget depicted in Fig.~\ref{fig:reduction}, where the weights of the edges are also depicted in the figure, and where we assume for example that $x_i$ appears negatively in $c_{i_1}$ and positively in $c_{i_2}$ and $c_{i_3}$.

\begin{figure}[h!]
\vspace{-.99cm}
    \center\includegraphics[width=0.97\textwidth]{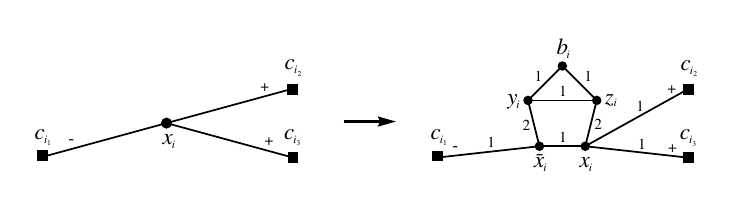}
    \vspace{-.65cm}
    \caption{\label{fig:reduction} Gadget of the reduction in the proof of Theorem~\ref{thm:NPh}.}
\end{figure}

Since each variable $x_i$ appears at most 3 times in~$\phi$, in the embedding of $G_{\phi}$ the edges between $x_i$ and the clause-vertices corresponding to positive or negative variables (namely, \emph{positive} and \emph{negative} edges) appear consecutively in all circular orderings around $x_i$. Thus, in the gadget, all positive (resp. negative) edges can indeed be incident with $x_i$ (resp. $\bar{x}_i$) while preserving planarity\footnote{Alternatively, the edge $\{x_i, \bar{x}_i\}$ can be seen as the result of subdividing the unique edge incident to~$x_i$ in~$G$ whose sign (positive/negative) differs from the other two.}

We now triangulate arbitrarily the graph obtained so far with edges of weight 5, except for the 4-cycle defined by the vertices $\{x_i, \bar{x}_i, y_i,z_i\}$ inside every gadget corresponding to a variable-vertex $x_i$. Finally, we add two adjacent new vertices $a_1$ and $a_2$ (the apices) and make them adjacent to all other vertices. The weights of these new edges are defined as follows. The edge between $a_1$ and $a_2$ has weight 5, and all other edges containing $a_2$ have weight 2. The edges between $a_1$ and the vertices $b_i$ have weight 1, and all other edges containing $a_1$ have weight 5. This completes the construction of $G$, which is clearly a 2-apex graph, as the removal of $a_1$ and $a_2$ yields a planar graph, relying on the planarity of $G_{\phi}$.

We claim that $G$ can be completed into a 2-apex graph $G'$ of diameter at most 4 if and only if $\phi$ is satisfiable. Note that the only edges that can be added to $G$ in order to build a 2-apex graph $G'$ are those inside the 4-cycles $\{x_i, \bar{x}_i, y_i,z_i\}$ for each variable-vertex $x_i$, namely either the edge $\{x_i,y_i\}$ or the edge $\{\bar{x}_i,z_i\}$ (see Fig.~\ref{fig:reduction}). Let us call these edges \emph{potential} edges. The choice of the edge $\{x_i,y_i\}$ (resp. $\{\bar{x}_i,z_i\}$) corresponds to setting the variable $x_i$ to true (resp. false) in the formula $\phi$. The reader can check that in $G$ all pairs of vertices are within weighted distance at most 4 (via $a_2$ or some~$b_i$) except for $a_1$ and the clause-vertices $c_j$. It holds that each clause-vertex $c_j$ is within distance at most 4 from  $a_1$ in $G'$ if and only if for at least one of the literals $\ell$ contained in the clause, the potential edge containing $\ell$ has been chosen in order to construct $G'$, which translates into the fact that the clause is satisfied. Therefore, $G$ can be completed into a 2-apex graph $G'$ of diameter at most 4 if and only if all clauses of $\phi$ can be simultaneously satisfied, as we wanted to prove. \end{proof}

\end{appendix}
\end{document}